\documentclass[preprint,12pt,english]{article}
\usepackage{amsmath,amsthm}
\usepackage[utf8x]{inputenc}
\usepackage{amssymb,textcomp,xcolor}
\usepackage{subfig}
\usepackage{subfig}
\usepackage{tikz}
\usepgflibrary{shapes.symbols} 
\usepgflibrary[shapes.symbols] 
\usetikzlibrary{shapes.symbols} 
\usetikzlibrary[shapes.symbols] 
\usepgflibrary{shapes.geometric} 
\usepgflibrary[shapes.geometric] 
\usetikzlibrary{shapes.geometric} 
\usetikzlibrary[shapes.geometric] 

\newcommand{\dahntab}[1]{
  \newbox\mybok%
  \setbox\mybok=\hbox{\vbox{
      \begin{tabbing}
        #1
      \end{tabbing}%
    }}

  \newdimen\bokwidth%
  \bokwidth=\wd\mybok%
  \newdimen\myl%
  \myl=\textwidth%
  \divide\myl by 2%
  \divide\bokwidth by -2%
  \advance\myl by\bokwidth%
  \vrule width\myl height 0pt depth 0pt%
  \usebox\mybok%
}
\newtheorem{observation}{Observation}
\newtheorem{definition}{Definition}
\newtheorem{theorem}{Theorem}
\newtheorem{proposition}{Proposition}
\newtheorem{claim}{Claim}
\newtheorem{lemma}{Lemma}
\newtheorem{expl}{Example}

\usepackage{natbib}
 \bibpunct[, ]{[}{]}{,}{n}{}{,}%

\usepackage[colorlinks=true,breaklinks=true,bookmarks=true,urlcolor=blue,
     citecolor=blue,linkcolor=blue,bookmarksopen=false,draft=false]{hyperref}

\def\EMAIL#1{\href{mailto:#1}{#1}}
\def\URL#1{\href{#1}{#1}}         

\begin{document}




\title{\textbf{\Large Minimum Entropy Submodular Optimization (and Fairness in Cooperative Games)}}
\author{Cosmin Bonchi\c{s}, Gabriel Istrate\footnote{Department of Computer Science, West University of Timi\c{s}oara and the e-Austria Research Institute,  Bd. V. P\^{a}rvan 4, cam 045B, Timi\c{s}oara, RO-300223, Romania. Corresponding author email: {\tt gabrielistrate@acm.org }}}

\maketitle
\begin{abstract}
We study minimum entropy submodular optimization, a common generalization of 
the minimum entropy set cover problem, studied earlier by Cardinal et al., 
and the submodular set cover problem 
(Wolsey \cite{wolsey-submodular}, Fujishige \cite{generalized-ssc}, etc).

We give a general bound of the approximation performance of the greedy algorithm using an approach that can be interpreted in terms 
of a particular type of biased network flows. As an application we rederive known results for the Minimum Entropy Set Cover 
and Minimum Entropy Orientation problems, and obtain a nontrivial bound for a new problem called the Minimum Entropy Spanning Tree problem.
 
The problem can be applied to (and is partly motivated by) the definition of worst-case approaches to fairness in concave cooperative games, similar 
 to the notion of price of anarchy in noncooperative settings. 
 
\end{abstract}

 
\section{Introduction}

Submodularity is a significant structural property of set functions, encoding the notion of {\em diminishing returns} and plays a crucial role in many scientific areas including combinatorial optimization \cite{fujishige2005submodular}, cooperative game theory \cite{shapley1971cores,branzei2008models}, information theory \cite{madiman2010information} and in applications involving {\em clustering} \cite{nagano2010minimum}, {\em learning} \cite{golovin2011adaptive}, {\em social and sensor networks} \cite{kempe2003maximizing}, {\em natural language processing} \cite{lin2010multi}, {\em signal processing} \cite{cevher2011greedy} or {\em constraint satisfaction} \cite{azar2011submodular}, to give just a few examples. Submodular function 
optimization is a well-established paradigm and reasonably well-understood: minimization has polynomial time algorithms \cite{schrijver2000combinatorial,iwata2003faster,iwata2009simple} while maximization is intractable. On the other hand for well-behaved submodular functions (so called integer polymatroids)
$f:\mathcal{P}(S)\rightarrow {\bf Z}_{+}$ finding the maximum is simple: the ground set set $S$ is a trivial solution. The interesting problem is finding a solution (a subset $A\subseteq S$ satisfying $f(A)=f(S)$) having minimum cost. This is essentially an instance of the {\em submodular set cover} problem \cite{fujita-ssc}. 

Minimizing the cost of $A$ is not the only possible objective function to investigate in this setting: a maximum likelihood approach to an inference problem in computational biology led Halperin and Karp \cite{halperin-karp-sc} to study a {\em minimum entropy} version of the set cover problem. Finding an optimal solution is in general NP-hard, but Halperin and Karp showed that the GREEDY algorithm produces an approximate solution whose entropy differs only by a constant factor to that of the optimal solution. A tight estimate was obtained by Cardinal et al. \cite{tight-minentropy-setcover} who subsequently studied other combinatorial problems under minimum entropy objectives \cite{cardinal2008minimum,cardinal2009minimum}. 

It must be stressed that minimizing entropy is an approach that goes beyond the problem studied by Halperin and Karp: for covering-type problems the connection between maximum-likelihood and minimum entropy is quite general. To give just one example, an even earlier problem that exploited the connection between maximum likelihood and minimum entropy  is  {\em word segmentation} \cite{wang2001minimum}. Minimizing entropic measures has other applications: for instance, in \cite{jajamovich2012maximum} the authors consider a sparse dictionary-based approach to maximum parsimony haplotype inference via minimizing a non-extensive variant of the entropy called {\em Tsallis entropy}. 

In this paper we unify these two directions, submodular optimization and combinatorial optimization under minimum entropy objective function, by investigating a minimum entropy version of the submodular set cover problem. While the problem is clearly NP-complete, our main result show that the approximation performance of the GREEDY algorithm can be quantified using a covering-like parameter that has an interpretation in terms of a type of certain ``biased'' network flows. This interpretation allows a fairly illuminating rederivation of results in \cite{cardinal2008minimum,cardinal2009minimum}. We then showcase the power of the method by providing an upper bound on the performance of the approximation performance of the greedy algorithm for a new problem, the {\em minimum entropy spanning tree problem}.

Besides the conceptual integration the framework we investigate was developed with several applications in mind. The most important of them (developed in a companion paper \cite{istrate-bonchis2012-tugames}) concerns the development of a worst-case approaches to fairness in concave cooperative games similar in spirit to the {\em price of anarchy} in noncooperative settings. The measure we propose are based on entropic concepts such as the Shannon divergence. We briefly outline this direction in Section~\ref{games}. Other potential applications arise in information theory \cite{smieja-tabor} and maximum-likelihood approaches to machine learning (in settings inspired by \cite{guillory2010interactive,golovin2011adaptive}).
 We plan to further explore and develop these connections in subsequent work. 

The plan of the paper is as follows: in Section 2 we briefly review some relevant concepts and notions. In Section 3 we point out the fact that the problems we are interested in are computationally intractable (NP-complete); we also introduce a greedy approach to minimum entropy submodular set cover. In section 4 we discuss an integer programming formulation of this problem. Section 5 contains our main result: we quantify the performance of the GREEDY algorithm with the help of a "covering constant" developed using the IP in Section 4. We then rederive (in Section 6) existing results 
on the performance of the GREEDY algorithm for the Minimum Entropy Orientations and the Minimum Entropy Set Cover problems \cite{cardinal2008minimum,cardinal2009minimum}. Section 7 contains an interpretation of the 
covering constant using network flows that allows us to tighten up our main theorem using a "multi-level" version of our covering constant. As an application we obtain a 
result on the approximability of the Minimum Entropy Spanning Tree problem. We also briefly discuss the intended application to cooperative game theory. 

\section{Preliminaries} 
\label{sec2}

We will assume general familiarity with submodular optimization, see e.g. \cite{fujishige2005submodular}. In particular a set function $f: \mathcal{P}(U)\rightarrow {\bf R}_{+}$ will be called {\em monotone} if $f(S)\leq f(T)$ whenever $S\subseteq T\subseteq U$,  {\em submodular} if $f(S)+f(T)\geq f(S\cup T)+f(S\cap T)$ for all $S,T \subseteq U$, {\em modular} if $f(S)+f(T)= f(S\cup T)+f(S\cap T)$ for all $S,T \subseteq U$, and {\em polymatroidal} if $f$ is monotone, submodular and satisfies $f(\emptyset)=0$. 

We will use {\em the Shannon entropy} of a distribution $P=(p_{i})_{i\in I}$, defined as $Ent(P)=-\sum_{i\in I} p_{i}\log_{2}(p_{i})$. 

An instance of the classical {\em (Minimum Cost) Set Cover} (SC) is specified by an universe $U$ and a family $\mathcal{P}=\{\mathcal{P}_{1},\ldots, \mathcal{P}_{m}\}$ of parts of $U$. Each set $\mathcal{P}_{i}$ comes with a 
nonnegative {\em cost} $c(i)$. The goal is to cover the set $U$ by a family of parts from $\mathcal{P}$ of minimal total cost. 

The following classical extension of the set cover problem \cite{fujita-ssc} shares many properties with problem SC.  

\begin{definition}\mbox{ } {\bf [SUBMODULAR SET COVER ] (SSC): }
\begin{enumerate}
\item\noindent[GIVEN:] A set $U$ and a monotone, submodular function $f:\mathcal{P}(U)\rightarrow {\bf Z}_{+}$ and a {\em cost function} 
$c:U\rightarrow {\bf R}_{+}$. The cost of a set $S$, denoted $c(S)$, is simply the sum of costs of its elements. Without lost of generality we can assume that $U=\{1, 2, \dots, m\}$ (also denoted [m]). Define also $N=f(U)$. 

\item\noindent[SOLUTIONS:] Subsets $S\subseteq U$ with $f(S)=f(U)$ (such a set is called {\em feasible}).
\item\noindent[OBJECTIVE:] To find a feasible subset $S\subseteq U$ of minimum cost. 
\end{enumerate}

\label{first-def}
\end{definition}

In particular, the performance of the Greedy algorithm for SSC was studied by Wolsey \cite{wolsey-submodular}, who showed that results for SC extend to this setup. Other generalizations were given, among other papers, in \cite{generalized-ssc}. 

For readers not familiar with SSC it is worth discussing the representation of problem SC as a special case of SSC, since a similar extension will motivate the technical definition of the problem we are interested in. 

Given any instance $(X,\mathcal{P})$ of SC of unit costs, define corresponding instance $(U,f)$ of SSC as follows:  
\begin{enumerate} 
 \item $U=\{1,2,\ldots, m\}$. 
\item For $S\subseteq U$ define $X_{S}=\bigcup_{i\in S}P_{i}$ and $f(S)=|X_{S}|$.  
\end{enumerate} 

It is well-known that function $f$ defined above is submodular. A set $S\subseteq U$ with $f(S)=f(U)$ corresponds to a family of parts $\{P_i\}_{i \in S}$  which cover $X$. 

Halperin and Karp introduced \cite{halperin-karp-sc} a variation of the SC problem that employs a different objective function: 

\begin{definition}\mbox{ } {\bf [MINIMUM ENTROPY SET COVER (MESC)]:}\\ 
 Let $X=\{x_1,x_2,\ldots, x_n\}$ for some $n\geq 1$ and $\mathcal{P}=\{P_{1},P_{2},\ldots, P_{m}\}$ be a family of subsets of $X$ which covers 
$X$. A {\em cover} is a function $g:X\rightarrow [m]$ such that for every $1\leq i \leq n$, 
\[
 x_{i}\in P_{g(i)}\mbox{(``$x_{i}$ is covered by set $P_{g_{i}}$'')}
\]

The {\bf entropy of cover g} is defined as 
\begin{equation} 
Ent(g)=-\sum_{i\in U} \frac{g(\{i\})}{g(U)}\ln\left(\frac{g(\{i\})}{g(U)}\right).
\label{entcover}
\end{equation} 
The objective of MESC is to find a cover $g$ of minimum entropy. 
\label{expl-sc}
\end{definition}

In the same way that problem SC was generalized to SSC, we extend problem MESC from Definition~\ref{first-def} to the following: 

\begin{definition}\mbox{ } {\bf [MIN-ENTROPY SUBMODULAR SET COVER] (MESSC): }
\begin{enumerate}
\item\noindent[GIVEN:] A set $U$ and a polymatroidal function $f:\mathcal{P}(U)\rightarrow {\bf Z}_{+}$. 
\item\noindent[SOLUTIONS: ] A {\em cover of f,} that is a modular function 
$g: \mathcal{P}(U)\rightarrow {\bf Z}_{+}$ with $g(U)=f(U)$ and $0\leq g(S)\leq f(S)$ for all $S\subseteq U$. 

The entropy of cover $g$ is defined as in equation~(\ref{entcover}).
\item\noindent[OBJECTIVE:]  Find a cover $g$ of $f$ of minimum entropy. 
\end{enumerate} 
\label{second-def}
\end{definition}




\section{Submodular Optimization with restrictions on solution structure: the Minimum Entropy Spanning Tree Problem} ~\\

Many submodular optimization problems arise from cooperative games on combinatorial structures \cite{bilbao2000cooperative}.
 In such games solutions are subject to  further constraints: A natural example is the setting where solution components form an independent set in a certain matroid. 

Games on matroids, or where the possible coalitions form a matroid have been thoroughly investigated in the literature (e.g. \cite{nagamochi1997complexity,bilbao2000cooperative,bilbao2001shapley,maffioli2007least}). One could naturally define a "min-entropy" version of minimum-base games on matroids. 
Such games essentially capture all instances of problem MESSC, as any integer polymatroid can be represented
using a set of flats in a certain matroid \cite{oxley2006matroid}.
 At this moment we are unable to deal with this problem in full generality. Instead we will concentrate on a special case, a combinatorial problem to which our main result will apply in a fairly nontrivial way. 

The problem we consider is a variant of {\em spanning tree games} a classical topic in the area of cooperative games on combinatorial structures (e.g. \cite{claus1973cost,bird1976cost,granot1981minimum,granot1984core,faigle1997complexity}):

\begin{definition}{\bf  [MIN-ENTROPY SPANNING TREE] (MEST):} 
\begin{enumerate}
\item\noindent[GIVEN:] A connected graph  $G=(V,E)$.
\item\noindent[SOLUTIONS: ] Given $S\subseteq E$, a {\em cover of $S$} is a function $u:S\rightarrow V$ such that for all $e\in S$, $u(e)$ is a vertex of $e$. 
\item\noindent[OBJECTIVES:] A spanning tree $T = (V, E(T)) \subseteq G$ and a cover $u$ of $E(T)$ that minimizes the entropy 
\[
Ent(T;u)=-\sum_{i\in V} \frac{|u^{-1}(\{i\}|}{|E(T)|}\log_{2}\left[\frac{|u^{-1}(\{i\}|}{|E(T)|}\right].
\]
Intuitively in MEST players correspond to nodes of the graph, each of which can only contribute (some of) the edges it is adjacent to, each edge at a unit cost. 
The goal of the players is to form a spanning tree with the contributed edges. We seek the ``most unbalanced'' (costwise) spanning tree. 
\end{enumerate} 
\label{def-mest}
\end{definition} 

Unlike many of the settings in the papers quoted above  
we allow a player to control a {\em set of edges}, rather than a single one. 
This particular choice ensures the fact that the cost function is submodular. In contrast, in the more classical variants of spanning tree games only a property weaker than submodularity called {\em permutational convexity} holds \cite{granot1982relationship}. 

Indeed, one can consider MEST as a problem with matroid restrictions on solution structure by considering the {\em cycle matroid} $M(G)$ of graph $G$, the matroid whose independent sets consist of sets $\mathcal{I}$ of edges of $G$ that do not contain a cycle. Bases in this matroid correspond to spanning trees of $G$. 
For all $S \subseteq V$ define
\begin{equation}
f(S)=\max\{|\mathcal{I}|:\mathcal{I}\in M(G), \forall e\in \mathcal{I}, e \mbox{ has at least one vertex in } S\}
\label{cyclematroid}
\end{equation}

\begin{lemma} 
Function $f$ from equation~(\ref{cyclematroid}) is submodular. 
\end{lemma}
\begin{proof} 
Define $g(S)$ to be the set of edges adjacent to $S.$ Function $g$ satisfies $g(S \cup T) = g(S) \cup g(T)$ and  $g(S \cap T) \subseteq g(S) \cap g(T).$
Let $r$ be the rank function of matroid $M(G).$ 
Clearly, $f(S) = r(g(S)).$ 

Appling the submodularity of the rank function $r$ to sets $g(S)$ and $g(T)$ we obtain: 
$f(S \cup T)+f(S \cap T) \leq r(g(S \cup T)) + r(g(S \cap T)) \leq r(g(S)) + r(g(T)) = f(S) +f(T)$
\end{proof}

\subsection{Computational intractability of problems MESSC and MEST.} 

Problems MESC, MESSC, MEST defined so far are combinatorial {\em optimization} problems. Turning them into decision problems is easy, though, in the standard manner: we just add an extra cost parameter $\lambda$ and ask to decide whether the given instance has a solution of cost at most $\lambda$. Without risk of confusion we will use the same name for the optimization problems and their corresponding decision variants. 

Since problem Minimum Entropy Set Cover Problem is NP-complete \cite{halperin-karp-sc}, so is its generalization MESSC. Theorem~\ref{thm:hard} shows that this is true for problem MEST as well, providing an alternate class of matroid restrictions (beyond those arising from set cover) for which the associated decision problem is hard. 

\begin{theorem}\label{thm:hard} 
Decision problem MEST is NP-complete. 
\end{theorem} 

The proof of Theorem~\ref{thm:hard} is given in the Appendix. 

\subsection{The Greedy Algorithm.} 

\begin{figure}[ht]
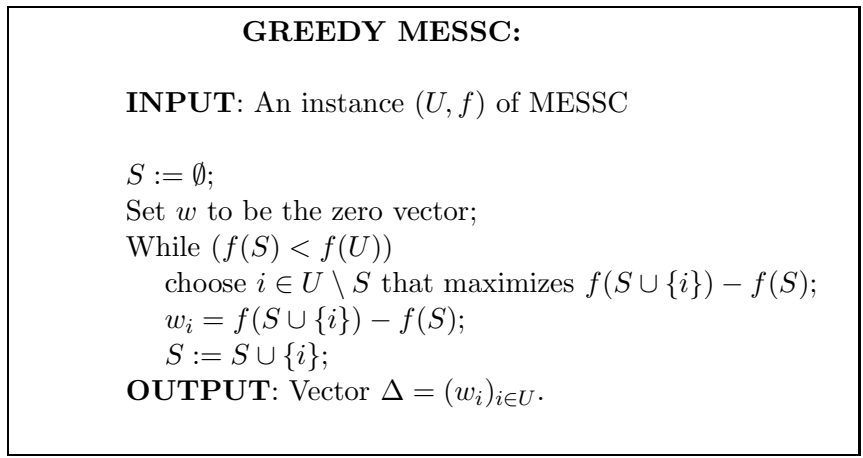

\begin{center}
\small{
\fbox{ 
\begin{minipage}[c]{0.8\columnwidth}%
\begin{center}
\dahntab{
\=\ \ \ \ \=\ \ \ \ \=\ \ \ \ \= 
{\bf GREEDY MESSC:} \\
\\
{\bf INPUT}: An instance $(U,f)$ of MESSC \\
\\
$S := \emptyset$; \\
Set $w$ to be the zero vector;\\ 
While ($f(S) < f(U)$) \\
\> \> choose $i\in U \setminus S$ that maximizes $f(S\cup \{ i \})-f(S)$;\\
\> \> $w_{i}= f(S\cup \{i\})-f(S)$;\\
\> \> $S := S\cup \{i\}$; \\
{\bf OUTPUT}: Vector $\Delta=(w_{i})_{i\in U}$. \\
}
\end{center}
\end{minipage}}}
\caption{Greedy algorithm for Minimum Entropy Submodular Set Cover.} 
\label{galg}
\end{center}
\end{figure}

Given the previous result, to solve problems MESSC or MEST we have to either resort to heuristic approaches or polynomial time  {\em approximation algorithms} \cite{vazirani-approx-book,williamson2011design}. In this paper we pursue the latter alternative. 

An approximation algorithm  based on the Greedy approach is presented in Figure~\ref{galg}. Note that it is well known that the resulting vector is a solution\footnote{This is easiest seen using the dual language of cooperative games. In so-called {\em concave games} (see Section~\ref{games} below) the core is non-empty \cite{shapley1971cores}, a polytope whose extremal points are those produced using a greedy approach on a given permutation of the elements of $U$. Our algorithm simply produces a particular such permutation.}. 

We will use, throughout the rest of the paper, the following notations: 
\begin{itemize} 
\item $i_{1},i_{2},\ldots, i_{l}$ will be the indices chosen by the GREEDY algorithm, in this order. Furthermore, we define for $1\leq r\leq l$ the {\em greedy rank function} by 
\[
rank(i_{r})=r.
\]
We extend function $rank$ to all elements of $U$ by considering an arbitrary (but fixed) ordering of such elements. 
\item For $1\leq r\leq l$, $W_{r}=\{i_{1},\ldots, i_{r}\}$ is the set of first $r$ elements added by the GREEDY algorithm; also $W_{0}=\emptyset$. $\Delta^{r}_{GREEDY}=w_{i_{r}}$ is the increase in the objective function caused by the $r$'th element chosen by $GREEDY$.   
\end{itemize} 

\begin{expl} 
 \mbox{ }Given a graph $G=(V,E),$ an instance of the problem MEST in Definition~\ref{def-mest}, one can inductively complete the 
GREEDY algorithm, constructing a solution $(S,u)$. $S$ is a set of edges of $G$; initially $S_{0}=\emptyset$. At stage $r,$ given the current set of edges $S_{r-1}$ 
constructed so far and the index $i_{r}$ chosen by the GREEDY algorithm, consider an independent set $I_{r}$ of 
cardinality $f(W_{r})$ defined in equation (\ref{cyclematroid}). Employing the exchange axiom of the cycle matroid we complete 
the independent set $S_{r-1}$ to an independent set $S_{r}$ by adding $f(W_{r})-f(W_{r-1})$ 
elements from $I_{r}$. Finally, extend $u$ to $S_r$ by defining $u(e) = i_{r}$ for all $e \in S_{r} \setminus S_{r-1}$ (all such edges are adjacent to $i_r$).
\label{greedy-memc}
\end{expl}

\section{Integer Programming Formulation and a Covering Coefficient.}

The minimum entropy set cover problem 
can be formulated as an integer programming problem (Figure~\ref{secondip} (a)). 
Similarly we can express MESSC (Figure~\ref{secondip} (b), with convention $0\log(0)=0$) by a rather artificial  integer program, whose usefulness will become clearer at a later stage. In a nutshell, the program provides a simple way to define the quantities measuring the performance of the greedy algorithm in our main result.

\begin{figure}[!h]
\begin{center}
\small{
\framebox{\begin{minipage}[t]{0.96\columnwidth}%

\subfloat[MESC]{%
\begin{minipage}[t]{0.45\columnwidth}%
\begin{align}
 & min\left[\sum_{S} - x_{S} \frac{|S|}{n}\log\left(\frac{|S|}{n}\right) \right] \\ 
 & \mbox{ s.t. } \notag\\
 & \sum_{S\subseteq T} x_{S}\geq 1, \mbox{ }T\in \mathcal{P}(U) \notag\\
 & x_{S}\in \{0,1\} \notag
\end{align}
\end{minipage}
}
~
\subfloat[MESSC]{%
\begin{minipage}[t][1\totalheight]{0.45\columnwidth}%
\begin{align}
 & \min\left[\sum_{i\in [m]}\sum_{\lambda=0}^{f(\{i\})} - x_{i,\lambda}\frac{\lambda}{n}\log\left(\frac{\lambda}{n}\right) \right] \label{ip-messc}\\ 
 & \mbox{ s.t. } \notag\\
 & \sum_{\lambda=0}^{f(\{i\})} x_{i,\lambda}= 1, \mbox{ }i\in [m] \notag\\
 & \sum_{i\in U-S}\sum_{\lambda=0}^{f(\{i\})} \lambda x_{i,\lambda}\geq f(U)-f(S), \mbox{ }S \subseteq [m] \notag\\
 & x_{i,\lambda}\in \{0,1\} \notag
\end{align}
\end{minipage}}
\end{minipage}}
\caption{\textbf{Minimum Entropy Integer Programming formulations}}
\label{secondip}
}
\end{center}
\end{figure}


\begin{proposition} 
 Given an instance of the MESSC, its solutions are in one-to-one correspondence to solutions of IP problem defined in Figure~\ref{secondip} (b). 
\end{proposition}
{\bf Proof. }
 Given a solution $z=(z_{j})_{j\in [m]}$ of MESSC define $\overline{x}_{j,\lambda}$ for $j\in [m]$ and $0\leq \lambda \leq f(\{j\})$ as follows:
\[
 \overline{x}_{j,\lambda}=\left\{\begin{array}{ll} 
  1 & \mbox{if } z_{j}=\lambda \\ 
   0 & \mbox{otherwise.}
  \end{array}
\right.
\]

\begin{claim} 
 If $z$ is a solution to MESSC then $\overline{x}$ is a feasible solution to system~(\ref{ip-messc}). 
\end{claim}
{\bf Proof.}
Since $z$ is a cover, $0\leq z_{j}\leq f(\{j\})$ for every $j\in [m]$ hence obviously  for any such $j$ there is exactly one $\lambda$, 
$0\leq \lambda\leq f(\{j\})$ with $\overline{x}_{j,\lambda}=1$. 

Thus $\displaystyle \sum \limits_{\lambda=0}^{f(\{j\})} \overline{x}_{j,\lambda}=1$ and $\displaystyle \sum \limits_{\lambda =0}^{f(\{j\})} \lambda \overline{x}_{j,\lambda}=z_{j}$.

From the definition of cover $\sum_{j\in S}z_j\leq f(S)$ and $\sum_{j\in U}z_j = f(U)$ therefore the second inequality follows.
\qed

Conversely, given a solution $w$ of ~(\ref{ip-messc}), define for $j\in [m]$ 
\begin{equation}\label{defX} 
X_{j}=\sum \limits_{\lambda=0}^{f(\{j\})} \lambda w_{j,\lambda}. 
\end{equation} 

\begin{claim} 
 $(X_{j})_{j\in [m]}$ defined above is a solution to MESSC whose entropy is equal to the value of the objective function of (\ref{ip-messc}) for $w$. 
\end{claim}
{\bf Proof. }
Equation~(\ref{defX}) and system~(\ref{ip-messc}) ensure the fact that $0\leq X_{j}\leq f(\{j\})$ for any $j\in [m]$, and 
$\sum_{j\in S} X_{j}\leq f(S)$ for any $S\subseteq [m]$, hence $X$ is a cover. 

Since for any $j\in [m]$ we have $\sum \limits_{\lambda=0}^{f(\{j\})}  w_{j,\lambda} =1$, exactly one such term in the above equality is 1, and the result immediately follows. 
\qed
\newline
\qed 

We now come up to a quantity that will play a fundamental role in our results below: for any $1 \leq r \leq l$ we define

\begin{equation}
a_{r}^j = f(W_{r}) - f(W_{{r}-1}) - \left(f(W_{r} \cup \{j\}) - f(W_{{r}-1} \cup \{j\}) \right). 
\label{a-r-j}
\end{equation}

The best way to make sense of the (admittedly unintuitive) definition of the $a_{r}^{j}$ coefficients above is to particularize them in the case of the set cover problem. Observation~\ref{a-mesc} below shows that in this case coefficients $a_{r}^{j}$ have a very intuitive description: they represent the size of the intersection of the $j$'th set $P_{j}$ to the $r$'th set in the GREEDY solution. 

\begin{observation} \label{a-mesc} 
\mbox{ } Let us consider the setting in Example~\ref{expl-sc}. Then 
\[
 a_{r}^{j}=|X_{W_{r}}\setminus X_{W_{r-1}}|-|X_{W_{r}}\setminus (X_{W_{r-1}}\cup P_{j})|=|(X_{W_{r}}\setminus X_{W_{r-1}})\cap P_{j}|
\]

\end{observation}

\begin{proposition}
  For any $1 \leq r \leq l$ and $1 \leq j \leq m$ we have $a_r^j \geq 0.$
\end{proposition}
When $j \in W_{r}$ this follows directly from the monotonicity of function $f.$
Assume now that $j \notin W_{r},$ employ the submodularity of function $f$ with $S = W_{r}$ and $T = W_{r-1} \cup \{j\}.$
\qed

\begin{figure}[!h] 
\begin{center}
\noindent \fbox{
\begin{minipage}[p]{0.80\columnwidth}%
\begin{align}
 & \min[\sum_{j\in [m]}\sum_{\lambda=1}^{f(\{j\})} - \frac{\lambda}{n}\log\left(\frac{\lambda}{n}\right) x_{j,\lambda}]\label{ipmessc2}& \\ 
 & \mbox{ s.t. } \notag\\
 & \sum_{\lambda=0}^{f(\{j\})} x_{j,\lambda}= 1, \mbox{ }j\in [m] \notag & \sum_{i \in U-S} X_{i}\geq f(U)-f(S), \mbox{ }S \subseteq [m] \notag \\
 & \sum_{\lambda=0}^{f(\{j\})} \lambda x_{j,\lambda}= X_{j}, \mbox{ } j\in [m] \notag  &  X_{j}= \sum_{r=1}^l Z_{r}^j, \forall j \in [m] \notag \\
 & x_{i,\lambda} \in \{0,1\} \notag & 0 \leq Z_r^j \leq a_r^j, Z_r^j\in {\bf Z} \notag
\end{align}
\end{minipage}} 
\end{center} 
\caption{Redundant IP formulation of MESSC} 
\label{redundant-ip}
\end{figure}

We will find it useful to introduce a large number of apparently superfluous variables in system~(\ref{ip-messc}) as 
presented in Figure~\ref{redundant-ip}. Intuitively $Z_{r}^j$ is the portion of optimal solution $X_{j}$ that can be 
assigned to cover ``the $r$'th set in the greedy solution''. This explains the newly introduced constraints: first, one has to allocate all of $X_{j}$ and no more than that. Second, one cannot allocate to any ``set $i_{r}$'' more than ``its intersection with $X_{j}$''. The quoted statements above make full sense, of course, only for regular set cover (Example~\ref{expl-sc})


To state the main result we need: 

\begin{definition} \label{Z_alpha_delta}
Given polymatroid $G$, let $\alpha= \alpha_G$ the smallest positive value such that there exists an optimal solution of system~(\ref{ip-messc}) 
that can be completed to a solution of system~(\ref{ipmessc2}) by defining $Z_{r}^{j}$ so that inequalities 
\begin{equation}
 \sum_{j=1}^{m} Z_{r}^{j} \leq \alpha \cdot \Delta_{r}^{GREEDY} 
\label{defalpha}
\end{equation}
hold for any $1\leq r\leq l$. 
\end{definition}

Given the discussion above, intuitively $\alpha$ is a covering coefficient. 
It measures the amount of inherent ``redundancy'' into coverings of the GREEDY solution by pieces obtained by breaking up the optimal solution. In this sense, measure $\alpha$ has a somewhat similar flavor to the {\em curvature of a submodular cost function} defined by Wan et al. in \cite{wan2010greedy}. Of course, there are some differences as well: the latter measure is relevant for {\em minimum cost} submodular optimization. It also does not directly involve the GREEDY solution (whereas, in the interest of tightness, our concept does).

\begin{proposition}\label{alpha-one}
The coefficient $\alpha$ that satisfies system~(\ref{defalpha}) is always 
greater or equal to  $1$. 
\end{proposition}
{\bf Proof.} 
Sum all equations~(\ref{defalpha}) for all $r=1,\ldots, l$. 

The left-hand side is 
\begin{equation} 
 \sum_{r=1}^{l} \left(\sum_{j=1}^{m} Z_{r}^{j}\right) = \sum_{j=1}^{m} \left(\sum_{r=1}^{l} Z_{r}^{j} \right) = \sum_{j=1}^{m} X_{j}= N.
\end{equation}

On the other hand the right-hand side is 
\begin{equation} 
 \alpha \cdot \sum_{r=1}^{l} \Delta_{r}^{GREEDY} \leq \alpha\cdot N,
\end{equation}

by the GREEDY algorithm. 
The result follows. 
\qed

\section{Main result.}

In this section we state and prove our main result

\begin{theorem}\label{thm-main} 
 Given a polymatroid $G=(U,f)$, the greedy algorithm produces  a solution $f_{GREEDY}$ to the instance $G$ of MESSC related to the optimal cover $f_{OPT}$ by relation:  

\begin{equation} 
 Ent[f_{GREEDY}]\leq \frac{1}{\alpha}\cdot [Ent[f_{OPT}]+\log_{2}(e)]+[1-\frac{1}{\alpha}]\log_{2}(n)
\end{equation} 
\end{theorem}

{\bf Proof. } 

Let $(X_j)_{j\in [m]}$ an optimal solution of the system from Figure (\ref{redundant-ip}) and $(y_i)_{i \in [m]}$ the solution generated by the greedy algorithm.

By the greedy choice we infer
 \[ y_{r} = f(W_{r-1}\cup \{i_r\}) - f(W_{r-1})\]
for any $1 \leq r \leq l,$ with $y_{i} = 0$ for other $i.$

We first prove several auxiliary results:

\begin{claim}\label{sumA_l_j}
For any $j \in [m]$ we have
\[
\sum_{r=1}^{l} a_{r}^j = f(\{j\}) 
\]
\end{claim}

{\bf Proof. }
By the definition of $a_{r}^j$:
\begin{eqnarray*}
 \sum_{r=1}^{l} a_{r}^j  & = & \sum_{r=1}^{l} \left( f(W_{r}) - f(W_{{r}-1}) - \left(f(W_{r} \cup \{j\}) - f(W_{{r}-1} \cup \{j\}) \right) \right) \\
 & = & f(W_{l}) - f(W_{0}) - \left(f(W_{l} \cup \{j\}) - f(W_{0} \cup \{j\}) \right)  \\
 & = & n - (n - f(\{j\}))  = f(\{j\}).
\end{eqnarray*}

The computations are justified by equalities $f(W_{0}) = 0$ and $f(W_{l} \cup \{j\} = f(W_{l}).$
\qed

On the other hand, we have:
\begin{claim}\label{sumA_r_j}
 For any $1 \leq r \leq l$ and $j \in [m]$ we have:
\[
f(\{j\}) - \sum_{k=1}^{r} a_{k}^j = f(W_{r} \cup \{j\}) - f(W_{r}).
\]
\end{claim}
{\bf Proof. }
\begin{eqnarray*}
 \sum_{k=1}^{r} a_{k}^j  & = & \sum_{k=1}^{r} \left( f(W_{k}) - f(W_{{k}-1}) - \left(f(W_{k} \cup \{j\}) - f(W_{{k}-1} \cup \{j\}) \right) \right) \\
 & = & \sum_{k=1}^{r} \left( f(W_{k}) - f(W_{{k}-1}) \right) - \sum_{k=1}^{r} \left( f(W_{k} \cup \{j\}) - f(W_{{k}-1} \cup \{j\}) \right) \\
 & = & f(W_{r}) - f(W_{0}) - \left(f(W_{r} \cup \{j\}) - f(W_{0} \cup \{j\}) \right)  \\
 & = & f(W_{r}) - f(W_{r} \cup \{j\}) + f(\{j\}). 
\end{eqnarray*}
\qed

\begin{claim}\label{ProdFact} 
Given coefficient $\alpha$ from the definition (\ref{Z_alpha_delta}) we have 
\[
\prod_{r=1}^{l} y_{r}^{\alpha y_{r}} \geq \prod_{j=1}^{m} X_j!
\]
\end{claim}

{\bf Proof. }
By the greedy choice, Claim (\ref{sumA_r_j}) and Claim (\ref{sumA_l_j})
\begin{eqnarray*}
y_{r}  & = & f(W_{r-1}\cup \{i_r\}) - f(W_{r-1}) \geq f(W_{r-1}\cup \{j\}) - f(W_{r-1}) \\
& = & f(\{j\}) -  \sum_{k=1}^{r-1} a_{k}^j  = \sum_{k=r}^{l} a_{k}^j 
\end{eqnarray*}
By the system in Figure (\ref{redundant-ip}) we have $Z_{r}^j \leq a_{r}^j$ and we obtain:
\[
 y_{r} \geq \sum_{k=r}^{l} Z_{k}^j = X_j - \sum_{k=1}^{r-1} Z_{k}^j
\]
Therefore,
\[
 \prod_{j=1}^{m}\left(\prod_{r=1}^{l} \left(y_{r}\right)^{Z_{k}^j} \right) \geq \prod_{j=1}^{m}  
\left( \prod_{r=1}^{l} \left(X_j - \sum_{k=1}^{r-1} Z_{k}^j \right)^{Z_{k}^j} \right) \ \geq \prod_{j=1}^{m} \left(X_j\right)!
\]

On the other hand, by Definition (\ref{Z_alpha_delta})
\[
 \prod_{j=1}^{m}\left(\prod_{r=1}^{l} \left(y_{r}\right)^{Z_{k}^j} \right)  = \prod_{r=1}^{l} \left(y_{r}\right)^{\sum_{j=1}^{m} Z_{k}^j} 
\leq \prod_{r=1}^{l} \left(y_{r}\right)^{\alpha y_{r}}.
\]
\qed

With Claim (\ref{ProdFact}) in hand, we get 
\begin{eqnarray*}
 ENT[f_{GREEDY}] & = & - \sum_{r=1}^{l}\frac{y_{r}}{n}\log_{2} \left(\frac{y_{r}}{n} \right)  =  - \sum_{r=1}^{l}\frac{y_{r}}{n}\log_{2} (y_{r}) +\log_{2}(n) \\
 & = & - \frac{1}{n} \frac{1}{\alpha} \sum_{r=1}^{l} \log_{2} y_{r}^{\alpha y_{r}} +\log_{2} n  =  - \frac{1}{n} \frac{1}{\alpha} \log_{2} \prod_{r=1}^{l}  y_{r}^{\alpha y_{r}} +\log_{2} n  \\
 & \leq & - \frac{1}{n} \frac{1}{\alpha} \log_{2} \prod_{j \in OPT}  X_j! +\log_{2}n 
\end{eqnarray*}

Using inequality $n! \geq \left(\frac{n}{e}\right)^{n}:$
\begin{eqnarray*}
  ENT[f_{GREEDY}] & \leq & - \frac{1}{n} \frac{1}{\alpha} \log_{2} \prod_{j \in OPT}  \left(\frac{X_j}{e}\right)^{X_j}  +\log_{2} n \\
& = & - \frac{1}{n} \frac{1}{\alpha} \sum_{j \in OPT} X_j \left( \log_{2} X_j - \log_{2} e \right)  +\log_{2}n \\
& = & \frac{1}{\alpha} \left(-\sum_{j \in OPT} \frac{X_j}{n} \log_{2} X_j \right) + \frac{1}{\alpha} \log_{2} e + \log_{2}n \\
& = & \frac{1}{\alpha} \left(-\sum_{j \in OPT} \frac{X_j}{n} \log_{2} \frac{X_j}{n} - \log_{2} n \right) + \frac{1}{\alpha} \log_{2} e + \log_{2}n \\
& = & \frac{1}{\alpha} \left( ENT[f_{OPT}] + \log_{2} e \right) + \left(1 - \frac{1}{\alpha} \right) \log_{2} n
\end{eqnarray*}

\qed

\section{Applications: Minimum Entropy Set Cover, Minimum Entropy Orientation.} 

A simple problem where one can determine the value of $\alpha$ is the Minimum Entropy Orientation 
(MEO) problem \cite{cardinal2008minimum,cardinal2009minimum}. Indeed, we recover (using a different method) the upper bound on the performance  
of the GREEDY algorithm for MEO (an algorithm that is, however, not optimal \cite{cardinal2008minimum}). The result will be generalized next to problem MESC (with an even  simpler proof). We have chosen to include it here, though, in this form as the proof  is going to be useful in the analysis of problem MEST. 


\begin{definition} {\bf [MIN-ENTROPY ORIENTATION (MEO)]:}
\begin{enumerate}
\item\noindent[GIVEN:] A graph $G=(V,E)$. 
\item\noindent[SOLUTIONS: ]
An {\em orientation of $G$} is a function $u:E\rightarrow V$ such that for all $e\in E$, $u(e)$ is one of the vertices of edge $e$.
\item\noindent[OBJECTIVE:] To find an orientation $u$ of $G$ that minimizes 
\[
Ent(S;u)=-\sum_{i\in [V} \frac{|u^{-1}(\{i\}|}{|E|}\log_{2}\left[\frac{|u^{-1}(\{i\}|}{|E|}\right].
\]
\end{enumerate} 
\end{definition} 

The MEO problem is a special case of MESC with sets corresponding to vertices and elements to edges adjacent to the given vertex.  Each instance $G=(V,E)$ of MEO can be regarded as an instance of MESC with submodular cost function 
\[
c(S)=|\{e\in E: E\cap S\neq \emptyset\}|.
\]

\begin{proposition} 
 For any instance $G$ of MEO $\alpha_G = 1$.
\end{proposition}


{\bf Proof.} A simple application of formula~(\ref{a-r-j}) yields
\begin{equation} 
 a_{r}^{j}=\left\{\begin{array}{l}
1,\mbox{ if } i_{r}\neq j, (i_{r},j)\in E, j\not\in W_{r} \\
\Delta_{r}^{GREEDY},\mbox{ if }i_{r}=j\\
0,\mbox{ otherwise.}\\
\end{array}
\right.
\end{equation}

This choice allows us to turn an orientation of minimum entropy (corresponding to 
an optimal solution $(X_{i})_{i\in [m]}$) into the greedy orientation as follows: 
\begin{itemize} 
 \item At each stage $r$, after choice of $i_{r}$ we reorient edges $(j,i_{r})$, $j\not \in W_{r}$ that have different orientations in the optimal and greedy solution. Correspondingly define $Z_{r}^{j}=1$ for such edges. 
\item Also let $Z_{r}^{i_{r}}$ be the number of edges $(j,i_{r})$ that are oriented towards $i_{r}$ in both the 
greedy and the optimal orientation. Note that there are at most $a_{r}^{i_{r}}=\Delta_{r}^{GREEDY}$ such edges.  
\item Note that an edge that is reoriented at stage $r$ is not reoriented again at a later stage (because of the 
restriction $j\not \in W_{r}$). Hence the process ends up with the greedy solution. In other words
\[
 \sum_{j} Z_{r}^{j}= \Delta_{r}^{GREEDY}.
\]
(as we add one unit for each edge counted by $\Delta_{r}^{GREEDY}$). 
\end{itemize}
\qed

We generalize the result above as follows: 

\begin{proposition}
 For any instance $G$ of MESC $\alpha_G = 1$. 
\end{proposition}
{\bf Proof. }
First remember that, as noted in Observation~\ref{a-mesc}, for MESC 
\begin{equation} 
a_{r}^{j}=|(X_{W_{r}}\setminus X_{W_{r-1}})\cap P_{j}|
\end{equation} 

 Let $u:[N]\rightarrow [m]$ be an optimal solution to MESC, i.e. for any $1\leq i\leq N$, $i\in P_{u(i)}$ and the cover specified by $u$ has minimum entropy. 

Denote, for $j=1,\ldots, m$, $U_{j}=\{x\in [N]:u(x)=j\}$. $U_{j}\subseteq P_{j}$ is the set of elements assigned by cover $u$ to set $P_{j}$. 

Define, for $1\leq r\leq l$
\begin{equation} 
Z_{r}^{j}=|U_{j}\cap (X_{W_{r}}\setminus X_{W_{r-1}})|.
\end{equation} 

Then $0\leq Z_{r}^{j}\leq a_{r}^{j}$. Moreover
\[
\sum_{r=1}^{l} Z_{r}^{j}=\sum_{r=1}^{l} |U_{j}\cap (X_{W_{r}}\setminus X_{W_{r-1}})|=|U_{j}|
\]
\[
\sum_{j=1}^{m} Z_{r}^{j}=\sum_{j=1}^{m} |U_{j}\cap (X_{W_{r}}\setminus X_{W_{r-1}})|=|X_{W_{r}}\setminus X_{W_{r-1}}|,
\]
(as each of the two state systems $(U_{j})_{j\in [m]}$ and $((X_{W_{r}}\setminus X_{W_{r-1}}))_{r=1}^{l}$ consists of disjoint sets) hence $X_{j}=|U_{j}|$ and $Z_{r}^{j}$ satisfy system in Figure~(\ref{redundant-ip}) and equation~(\ref{defalpha}) with $\alpha=1$. 
\qed

\section{Network flow interpretation and extension of the main result.} 

Sometimes the application of our main result to specific problems is not quite as straightforward as above. An example is the Minimum Entropy Spanning Tree problem from Definition~\ref{def-mest}. Intuitively, in this case we would also like to apply Theorem~\ref{thm-main} by proving that $\alpha_G=1$. However, this second result is not easy to obtain. To understand why, we will first reinterpret constant $\alpha_G$ using network flows. This will allow us to eventually define a related  
constant $\beta_G$. The difference between $\alpha$ and $\beta$ is that roughly $\alpha$ is defined in terms of one-stage network flows, whereas in $\beta$ we will allow multistage constructions. 

We will prove a variant of Theorem~\ref{thm-main} for constant $\beta$, then we will give a multistage flow construction witnessing that for any instance $(G,w)$ of MEST $\beta=1$. This will prove the desired result. 

\begin{expl}
\mbox{ }For the MEST problem we have 
\[
f(W_{r})-f(W_{r-1})=|\{e\in E(G): e=(i_{r},k), k\sim i_{r}, k\not\sim W_{r-1}\}|
\]
Similarly 
\[
f(W_{r}\cup \{j\})-f(W_{r-1}\cup \{j\})=|\{e\in E(G): e=(i_{r},k), k\sim i_{r},k\not\sim W_{r-1}\cup \{j\} \}|
\]
Therefore 
\begin{equation} 
 a_{r}^{j}=\left\{\begin{array}{l}
1,\mbox{ if } i_{r}\neq j, i_{r}\sim j, j\not\in W_{r} \\
|\{k:k\sim i_{r},k\sim j, k\not\sim W_{r-1}\}|,\mbox{ if } i_{r}\neq j, i_{r}\not \sim j, j\not\in W_{r} \\
\Delta_{r}^{GREEDY},\mbox{ if }i_{r}=j\\
0,\mbox{ otherwise.}\\
\end{array}
\right.
\end{equation}
\end{expl} 

\begin{figure} 
\begin{center} 
 \begin{tikzpicture}
[place/.style={circle,draw=blue!50,fill=blue!20,thick,inner sep=1pt,minimum size=5mm},
transition/.style={rectangle,draw=black!50,fill=black!20,thick},
pre/.style={<-,shorten <=2pt,semithick},
post/.style={->,shorten >=2pt,semithick}]
\node[transition] (start) at 	(0,0) {$s$};
\node[place] (y1) at 	(2,-1.5) {$x_m$};
\node[] (y2) at 	(2,-0.5) {$\vdots$};
\node[place] (yp) at 	(2, 0.5) {$x_2$};
\node[place] (yl) at 	(2, 1.5) {$x_1$};
\foreach \x in {y1, y2, yp, yl}
	\draw[post] (start) to (\x);
\node[place] (x1) at 	(6,-1.5) {$y_{l}$};
\node[] (x2) at 	(6,-0.5) {$\vdots$};
\node[place] (xp) at 	(6, 0.5) {$y_2$};
\node[place] (xm) at 	(6, 1.5) {$y_1$};

	\draw[post] (y1) to (xm);
	\draw[post] (y2) to (x1);
	\draw[post] (yl) to (xp);

\node[transition] (end) at (8,0) {$t$};

\foreach \x in {x1, x2, xp, xm}
	\draw[post] (\x) to (end);
\end{tikzpicture}
\end{center}
\caption{Network flow interpretation of constant $\alpha$ in Definition~\ref{Z_alpha_delta}.}
\label{first-flow} 
\end{figure}
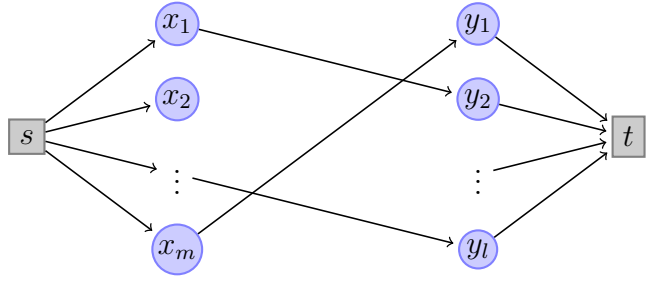

\begin{figure} 
\begin{center} 
\begin{tikzpicture}
[place/.style={circle,draw=blue!50,fill=blue!20,thick,inner sep=1pt,minimum size=5mm},
transition/.style={rectangle,draw=black!50,fill=black!20,thick},
pre/.style={<-,shorten <=2pt,semithick},
post/.style={->,shorten >=2pt,semithick}]

\node[transition] (start) at 	(0,0) {$s$};
\node[place] (y1) at 	(2,-1.5) {$x_m$};
\node[] (y2) at 	(2,-0.5) {$\vdots$};
\node[place] (yp) at 	(2, 0.5) {$x_2$};
\node[place] (yl) at 	(2, 1.5) {$x_1$};

\node[place] (z1) at 	(4,-1.5) {$z_s$};
\node[] (z2) at 	(4,-0.5) {$\vdots$};
\node[place] (zp) at 	(4, 0.5) {$z_2$};
\node[place] (zm) at 	(4, 1.5) {$z_1$};

\foreach \x in {y1, y2, yp, yl}{
	\draw[post] (start) to (\x);
}

\draw[post] (y1) to (zm);
	\draw[post] (y2) to (z1);
	\draw[post] (yl) to (zp);


\node[place] (x1) at 	(6,-1.5) {$y_l$};
\node[] (x2) at 	(6,-0.5) {$\vdots$};
\node[place] (xp) at 	(6, 0.5) {$y_2$};
\node[place] (xm) at 	(6, 1.5) {$y_1$};

\node[transition] (end) at (8,0) {$t$};

\draw[post] (z1) to (x1);
	\draw[post] (z2) to (xm);
	\draw[post] (zp) to (xm);

\foreach \x in {x1, x2, xp, xm}{
	\draw[post] (\x) to (end);
}
\end{tikzpicture}
\end{center}
\caption{Multistage flow network between solutions}  
\label{second-flow}
\end{figure}
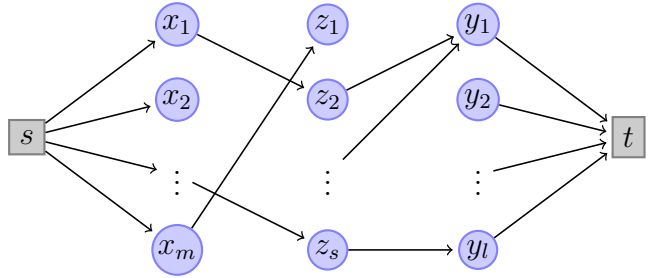

Consider, though, the flow network in Figure~\ref{first-flow}. In addition to source/sink nodes $s,t$, 
$F$ has two layers of nodes;  the first layer 
of nodes corresponding to the optimal solution, the second layer of nodes corresponding to the greedy one. In 
each layer we have a node for every player in the game. Edges appear between nodes of type $k$ and $i_{r}$, with 
capacity equal to $a_{k}^{r}$. The fact that the first layer of nodes corresponds to the optimal solution is reflected by setting capacity $X_{j}^{OPT}$ on the edge between node $s$ and node $j$. 
Similarly, capacities between node $i_{r}$ of the second layer and node $t$ are set to value $\Delta_{r}^{GREEDY}$. These values are seen as {\em requests} of node $Y_{r}$ that may be satisfied by the flow (which in general might send an amount larger than $\Delta_{r}^{GREEDY}$ to this node)

It follows that $\alpha_G=1$ amounts to the existence of a flow $f$ of value $n$ in the flow network of Figure~\ref{first-flow} (that is, $f$ satisfies the request of each node $Y_{r}$ exactly). More generally, $\alpha$ {\em is the minimum amount needed to multiply the capacities on the edges entering the sink $t$ in order to accommodate a flow with value $n$ from $s$ to $t$.}

Our solution to problem MEO could be easily recast in terms of flows: we construct the flow iteratively, by considering the paths between a node in the first layer and a node in the second layer inductively, in an order determined first by the order of second-layer nodes corresponding to the GREEDY algorithm and then going on nodes in the first layer according to a fixed ordering. 

There are lessons to be learned from the construction this flow and our proof of 
Theorem~\ref{thm-main}: The key point was that when we had to reorient an edge towards a node in the greedy solution, {\em we could do so without overflowing this node}. 
Similarly, the general proof depended on the following the invariant we maintained 
\begin{equation} \label{eq}
y_{r} \geq \sum_{k=r}^{l} Z_{k}^j
\end{equation} 
Condition~(\ref{eq}) does not have a direct flow interpretation, since $y_{r}$ is the request, rather than the actual flow value at the given node. However, its relaxation involving $\alpha$ does: the actual flow into node $y_{r}$ is at most $\alpha y_{r}$, so 
requiring that the total flow into node $y_{r}$ is at least $\sum_{k=r}^{l} Z_{k}^j$ guarantees the following relaxed version of equation~(\ref{eq}): 
$\alpha \cdot y_{r} \geq \sum_{k=r}^{l} Z_{k}^j$. We will see (in Proposition~\ref{thm-main-two} below) that the relaxed condition can be applied as well. 

We will generalize the setting of Theorem~\ref{thm-main} by considering flow networks with $q\geq 1$ levels (see 
Figure~\ref{second-flow} for $q=2$). The nodes in each level are ordered according to a fixed ordering, the same for all levels, say the ordering induced by the GREEDY algorithm, with nodes not chosen by this algorithm coming after all chosen nodes in a fixed, arbitrary sequence. Capacities in this network 
correspond either to values $a_{j}^{r}$ (if the chosen indices are $j$ and $i_{r}$, respectively) or $\infty$, for edges between nodes with the same index $j$ but on different levels. Note that each path ending in a greedy node with index $i_{r}$ has finite capacity, at most the capacity of its last edge. We will use notation $P:[j\ldots k]$ to indicate the fact that path $P$ starts at node $j$ on the first level and ends at node $k$ on the last level. We will also use notation $P\sim v$ to indicate the fact that path $P$ is adjacent to node $v$.

We will consider a total ordering $<$ on paths which will be explicitly constructed later. 

\begin{definition}
A flow $f$ is {\em admissible with respect to total path ordering $<$} if for any path $P$ between, say, node $X_{j}$ and $Y_{r}$, the remaining flow into node $X_{j}$ just before path $P$ is considered is at most the final value of the final total flow into node $Y_{r}$. 
Formally
\begin{equation}\label{admiss}
\sum_{Q\sim X_{j},P\leq Q} f(Q)\leq \sum_{W\sim Y_{r}} f(W).
\end{equation} 
\end{definition} 

Consider now a multiple-layer flow network corresponding to the optimal and greedy solutions, respectively (that is, the capacities of edges from s/into t are determined by the values of these solutions). Even with multiple layers it might not be possible to obtain an admissible  flow of value $n$. As before, the solution is to multiply the capacities of edges leading into node $t$ by some fixed amount $\beta$. 

\begin{definition}
Define $\beta_G$ as the infimum (over all multi-level flow networks corresponding to the optimal and greedy solution) of 
all values $\beta >0$ for which there exists a path ordering $<$ and a flow $f$ admissible w.r.t. $<$ such that for every pair of nodes $j$ and $r$, 
\begin{equation}\label{bflow} 
\sum_{j} \left(\sum_{P:[j\ldots i_{r}]} f_{P}\right)\leq \beta \cdot \Delta_{r}^{GREEDY}.
\end{equation} 
\label{beta-flow} 
The reader is requested to compare Definitions~\ref{Z_alpha_delta} and this definition. 
\end{definition}

Similarly to the proof of Proposition~\ref{alpha-one}, we obtain $\beta=\beta_G\geq 1$ always. On the other hand, admissibility will guarantee in general 
only a weaker version of inequality~(\ref{eq}): an extra $\beta$ factor is needed on the 
left-hand side (though this is not going to be weaker for the main setting we have in mind, $\beta=1$).

With this discussion we generalize Theorem~\ref{thm-main} as follows: 

\begin{proposition}\label{thm-main-two} 
 Given an instance $G=(N,f)$, of MESSC the greedy algorithm produces a solution $f_{GREEDY}$ satisfying 

\begin{equation} 
 \beta \cdot Ent[f_{GREEDY}]\leq Ent[f_{OPT}]+\log_{2}(e)+\beta\log_{2}(\beta)+(\beta-1)\log_{2}(n).
\end{equation} 
\end{proposition}


The proof is almost identical to that of Theorem~\ref{thm-main}: Let $(X_i)_{i\in [m]}$ an optimal solution of the system from Figure (\ref{redundant-ip}) and $(y_i)_{i \in [m]}$ the solution generated by the greedy algorithm.

Consider a multi-layer flow network such that equation~(\ref{bflow}) is satisfied with $\beta=\beta_G+\epsilon$, for some 
$\epsilon >0$. Let $f$ be the corresponding admissible flow and define $Z_{r}^{j}=\sum_{P:[j\ldots i_{r}]} f_{P}$.

\begin{claim}\label{ProdFact2}
We have
\[
\prod_{r=1}^{l} [(\beta_G+\epsilon)\cdot y_{r}]^{(\beta_G+\epsilon) y_{r}} \geq \prod_{j \in OPT} X_j!.
\]
\end{claim}

As flow $f$ is admissible, 

\[
 (\beta_G+\epsilon)\cdot y_{r} \geq \sum_{k=r}^{l} Z_{k}^j = X_j - \sum_{k=1}^{r-1} Z_{k}^j.
\]

This follows from considering the situation just before setting the flow on the lexicographically smallest path between node $j$ and $i_{r}$: $X_j - \sum_{k=1}^{r-1} Z_{k}^j$ is the amount of unsent flow at node $j$, to be sent on a path to one of nodes $i_{r},\ldots, i_{l}$. No such path has been considered yet, as they are lexicographically larger. 

Therefore, 
\[
 \prod_{j \in OPT} \prod_{r=1}^{l} \left((\beta_G+\epsilon)y_{r}\right)^{Z_{k}^j} \geq \prod_{j \in OPT} 
\left( \prod_{r=1}^{l} \left(X_j - \sum_{k=1}^{r-1} Z_{k}^j \right)^{Z_{k}^j} \right) \ \geq \prod_{j \in OPT}\left(X_j\right)!
\]

On the other hand, by Definition (\ref{beta-flow}) 
\[
\prod_{j \in OPT} \prod_{r=1}^{l} \left((\beta_G+\epsilon)y_{r}\right)^{Z_{k}^j} = \prod_{r=1}^{l} {(\beta_G+\epsilon)y_{r}}^{\sum_{j} Z_{r}^j} 
\leq \prod_{r=1}^{l} \left((\beta_G+\epsilon)y_{r}\right)^{(\beta_G+\epsilon)y_{r}}.
\]

\qed
\newline
The rest of the computation is similar, except that we also have to take the limit $\epsilon \rightarrow 0$ in the end, to obtain 
the desired result. 
\qed

\section{Application to the minimum entropy spanning tree problem.} 

Finally we are in position to  apply the result in Proposition~\ref{thm-main-two} by proving 
the following: 

\begin{theorem}\label{mest}
 For any instance $G$ of MEST, $\beta_G = 1$. That is, if $f_{GREEDY}$ is the cover provided by the GREEDY algorithm and $f_{OPT}$ is the optimal cover
\[
Ent[f_{GREEDY}]\leq Ent[f_{OPT}]+\log_{2}(e).
\]
\end{theorem}
{\bf Proof. }
We will create a flow, admissible with respect to some total path ordering $<$, that will witness the fact that $\beta_G=1$. To do so we first revisit the GREEDY algorithm for 
MEST (Example~\ref{greedy-memc}). 

As discussed there, the GREEDY algorithm builds an ``independent set'' (forest, in the particular case of MEST) incrementally: edges are only added, but not removed. After some 
stage $k$, $1\leq k\leq l$ the edges added by the GREEDY algorithm connect nodes in $W_{k}$ to some other nodes. Denote by $\delta(W_{k})$ the set of nodes not in $W_{k}$ but adjacent to some node in $W_{k}$ (after stage $k$). 

Consider some stage $r$, $1\leq r\leq l$. Denote by $C_{1},C_{2},\ldots, C_{p}$ the connected components (trees) created by the GREEDY algorithm after stage $r-1$.
 Node $i_{r}$ chosen at stage $r$ will connect to some of its adjacent nodes 
(not already selected), so that the resulted induced graph contains no cycles.

We infer the following: 
\begin{itemize} 
\item Any edge $(i_{r},x)$, with $x$ not in $C_{1}\cup C_{2}\cup \ldots \cup C_{p}$ is added by the GREEDY algorithm (charged to $i_{r}$) at stage $r$. 
\item The GREEDY algorithm also adds some of the edges $(i_{r},x)$, where $x$ 
belongs to a connected component among $C_{1},C_{2},\ldots, C_{p}$. It can only add such
 an edge if $i_{r}$ is not already connected to some node in that component 
(necessarily a member of $W_{r-1}$), thus creating a cycle. More precisely, in this case 
it will add {\em exactly one edge for each such component it's adjacent to}, merging in 
effect these components.

Even in this case, note that $x$ cannot belong to $W_{r-1}$, but to $\delta(W_{r-1})$. 
Indeed, suppose $x$ were an element in $W_{r-1}$. Edge $(i_{r},x)$ was not added when 
 $x$ was considered because it was creating a cycle. But then adding it would 
create a cycle now as well. 
\end{itemize} 

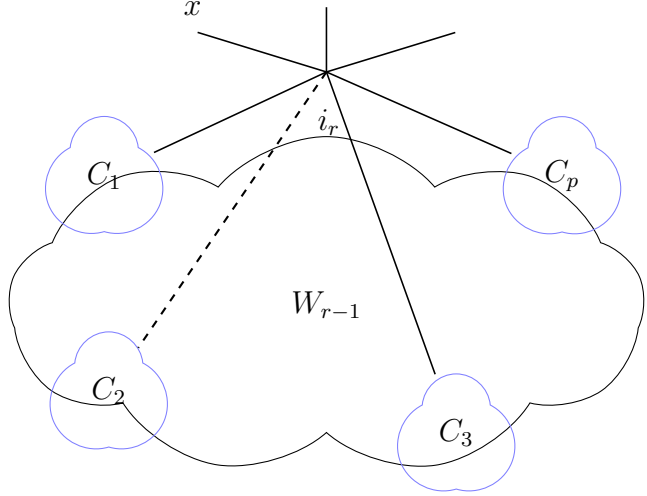
\begin{figure}[ht]
\begin{center} 
\begin{tikzpicture}
[comp/.style={cloud,draw=blue!50, cloud puffs=3, aspect=0.7, cloud puff arc=170,inner ysep=0.3cm},
place/.style={thick,inner sep=1pt,minimum size=5mm, right},
link/.style={-,shorten >=2pt,semithick, above=10}]

\node[name=G, shape=cloud, aspect=2, inner ysep=1cm] {};
\node[place] (ir) at (G.mid) [right=1, above=2] {$i_r$};
\node[name=w, shape=cloud, draw, cloud puffs=9, aspect=2.5, cloud puff arc=100, inner ysep=1cm] at (G.south) {$W_{r-1}$};
\node[comp] (c1) at (w.puff 2)[left=-15] {$C_1$};
\node[comp] (c2) at (w.puff 4)[right=0.02] {$C_2$};
\node[comp] (c3) at (w.puff 6)[left=-8, above=-9] {$C_3$};
\node[comp] (c4) at (w.puff 9)[right=-10] {$C_p$};

\foreach \k in {1,3,4}
  \draw[link] (G.mid) + (0,1) to (c\k);
\draw[dashed,thick] (G.mid) + (0,1) to (c2);

\foreach \k in {1,2,10}
  \draw[link] (G.mid) + (0,1) to (G.puff \k);

\node[place] (ir) at (G.puff 2) [above=1] {$x$};
\end{tikzpicture}
\caption{The GREEDY algorithm for MEST: at stage $r$, it adds edges from $i_{r}$ to nodes in components 
$C_{1},C_{3},C_{p}$ (merging them). It does {\em not} add an edge to component $C_{2}$, 
as there already existed a (shaded) edge from $i_{r}$ to a node in $C_{2}\cap W_{r-1}$. It 
also adds edges $(i_{r},x)$ to nodes $x$ outside of $C_{1},C_{2},\ldots, C_{p}$.}
\end{center} 
\end{figure}

As a consequence of the previous analysis the following holds: 
\begin{lemma} 
Suppose edge $(i_{r},x)$ is added by the GREEDY algorithm at stage $r$. Then 
\[
rank(x)>r, 
\]
where $rank(x)$ is the GREEDY rank of the element $x$, the stage when the element $x$ was chosen.
\end{lemma} 

Element $x$ clearly cannot belong to $W_{r-1}$, if $x$ falls in the first case of the 
previous discussion. As for the second case, by the argument there $x\in \delta(W_{r-1})$, which implies the fact $rank(x)>r$. 
\qed

We will also need a flow property  that ensures flow admissibility in the particular case
 of the MEST problem: 

\begin{definition} 
A flow is {\em biased (with respect to vertex ordering $r$)} if, for all nodes $j,l$
\begin{equation} 
\exists P:[j,l], f_{P}>0 \Rightarrow [rank(j)\geq rank(l)].
\label{biased-flow}
\end{equation} 
\end{definition} 

The importance of this notion lies in the fact that, while condition~(\ref{biased-flow}) is not necessarily satisfied ``between the endpoints'' of a flow, biased flows can intuitively be ``composed'', as rank inequality is transitive. 

Next we prove the following claim:
\begin{lemma}\label{first-lemma} 
Given an instance $X$ of the MEST problem let $X_{OPT}$ and $Y_{GREEDY}$ be the vectors corresponding to the optimal and greedy solution, respectively, with elements ordered according to the ordering induced by the greedy algorithm. 

Then there exists a biased flow $f$ with initial values $X_{OPT}$  and final values $Y_{GREEDY}$. 
\end{lemma} 
Let $T_{OPT}$ be the spanning tree (with oriented edges) corresponding to $X_{OPT}$, and let $T_{GREEDY}$ be the spanning tree with oriented edges corresponding to $Y_{GREEDY}$. We will construct a multi-level flow network and a greedy flow in stages, corresponding to edge moves that transform $T_{OPT}$ into $T_{GREEDY}$. Flow values on some node $v$ on an intermediate level correspond to edges oriented towards $v$ at that stage. 

Allowed moves are of two basic types: 
\begin{enumerate}
\item{\bf ``Edge reversals''.} Consider an edge $e=(w_{1},w_{2})$ in the current tree, oriented towards $w_{2}$. We reorient edge $e$ towards $w_{1}$. Biased edge reversals are those with $rank(w_{1})< rank(w_{2})$. 


\begin{figure}[h!]
 \begin{center}
\begin{tikzpicture}
[pre/.style={<-,shorten <=2pt,semithick},
place/.style={thick,inner sep=1pt,minimum size=5mm, right},
post/.style={->,shorten >=2pt,semithick},
arrow/.style={->, aspect=4, double, shorten >=4pt, thick},
elips/.style={draw=blue!40, fill=gray!25, shape=ellipse, inner ysep=1.6cm, inner xsep=1cm},
link/.style={-,shorten >=0.1pt,semithick}
]

\node[name=R,shape=rectangle, inner xsep=3.5cm] {};

\node[name=T,elips] at (R.west)[left=15] {};
\node[place] (w1) at (T.mid) [left=10, below=30] {$w_1$};
\node[place] (w2) at (T.mid) [right=7, above=35] {$w_2$};
\draw[post] (w1) to node[auto]{e} (w2);

\node[name=T1,elips] at (R.west)[right=15] {};
\node[place] (w1) at (T1.mid) [left=10, below=30] {$w_1$};
\node[place] (w2) at (T1.mid) [right=7, above=35] {$w_2$};
\draw[pre] (w1) to node[auto]{e} (w2);

\draw[arrow] (T) + (1.6,0) to (T1);

\matrix[column sep=1cm] at (R.east)
{
\node {.}; & \node{.}; \\
\node {.}; & \node{.}; \\
\node {.}; & \node{.}; \\
\node {.}; & \node (fw1){}; \\
\node {.}; & \node{.}; \\
\node {.}; & \node{.}; \\
\node (fw2){}; & \node{.}; \\
\node {.}; & \node{.}; \\
\node {.}; & \node{.}; \\
\node {.}; & \node{.}; \\
};
\node (ww2) at (fw2)[left] {$w_2$};
\node (ww1) at (fw1)[right] {$w_1$};
\draw[link] (fw2)circle(0.7mm) to (fw1)circle(0.7mm);

\end{tikzpicture}

\caption{Edge reversals and associated flow.}
\label{reversals}
\end{center} 
\end{figure}
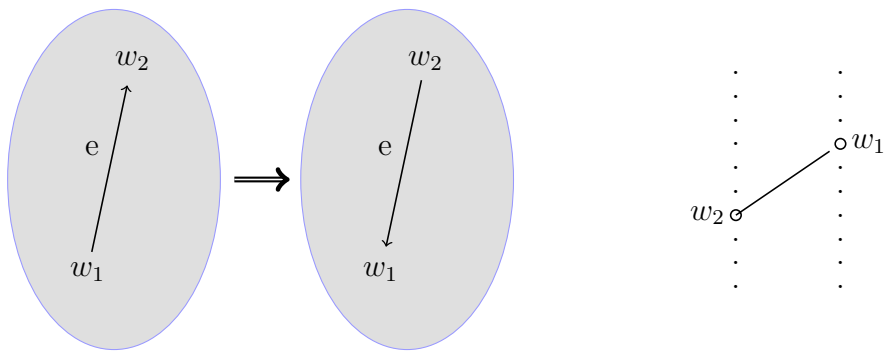


\item{\bf ``Rotations".} Consider an edge $e=(w_{1},w_{2})$ in the current tree, oriented towards $w_{1}$, and let $w_{3}$ be another vertex connected to $w_{1}$, such that edge $(w_{1},w_{3})$ is {\em not} in the tree. Replace $(w_{1},w_{2})$ by $(w_{1},w_{3})$. 


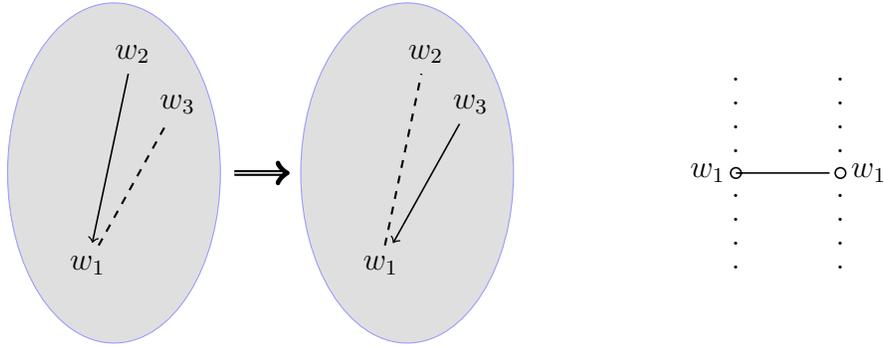
\begin{figure}[h!]
\begin{center} 
\begin{tikzpicture}
[pre/.style={<-,shorten <=1pt,semithick},
place/.style={thick,inner sep=1pt,minimum size=5mm, right},
post/.style={->,shorten >=1pt,semithick},
arrow/.style={->, aspect=4, double, shorten >=4pt, thick},
elips/.style={draw=blue!40, fill=gray!25, shape=ellipse, inner ysep=1.6cm, inner xsep=1cm},
link/.style={-,shorten >=1pt,semithick}
]
\node[name=R,shape=rectangle, inner xsep=3.5cm] {};
\node[name=T,elips] at (R.west)[left=15] {};
\node[place] (w1) at (T.mid) [left=10, below=30] {$w_1$};
\node[place] (w2) at (T.mid) [right=7, above=35] {$w_2$};
\node[place] (w3) at (T.mid) [right=24, above=16] {$w_3$};
\draw[post] (w2) to (w1);
\draw[dashed,thick] (w1) to (w3);

\node[name=T1,elips] at (R.west)[right=15] {};
\node[place] (w1) at (T1.mid) [left=10, below=30] {$w_1$};
\node[place] (w2) at (T1.mid) [right=7, above=35] {$w_2$};
\node[place] (w3) at (T1.mid) [right=24, above=16] {$w_3$};
\draw[post] (w3) to (w1);
\draw[dashed,thick] (w1) to (w2);

\draw[arrow] (T) + (1.6,0) to (T1)[left=3];

\matrix[column sep=1cm] at (R.east)
{
\node {.}; & \node{.}; \\
\node {.}; & \node{.}; \\
\node {.}; & \node{.}; \\
\node {.}; & \node{.}; \\
\node (fw2){}; & \node (fw1){}; \\
\node {.}; & \node{.}; \\r(
\node {.}; & \node{.}; \\
\node {.}; & \node{.}; \\
\node {.}; & \node{.}; \\
};
\node (ww2) at (fw2)[left] {$w_1$};
\node (ww1) at (fw1)[right] {$w_1$};
\draw[link] (fw2)circle(0.7mm) to (fw1)circle(0.7mm);
\end{tikzpicture}

\caption{Edge rotations and associated flow.}
\label{rotations}
\end{center} 
\end{figure} 


\hspace{-10mm} We will use, in fact, a third type, specified as follows, composed of the previous two moves. 

\item{\bf ``Edge slidings".} Let $a,b,c$ be three nodes 
Assume that edge $(b,c)$ is in the current tree (oriented towards $b$) and edge $(a,b)$ is not. Replace edge $(b,c)$ by edge $(a,b)$, oriented towards $a$ by first doing a ``rotation'' (move of type 2) around node $b$, and then reorienting edge $(a,b)$ towards $a$ (move of type 1). Biased edge slidings are those corresponding to the case $rank(a)<rank(b)<rank(c)$. 

 
\begin{figure}[h!]
\begin{center}
\begin{tikzpicture}
[pre/.style={<-,shorten <=1pt,semithick},
place/.style={thick,inner sep=1pt,minimum size=5mm, right},
post/.style={->,shorten >=1pt,semithick},
arrow/.style={->, aspect=4, double, shorten >=4pt, thick},
elips/.style={draw=blue!40, fill=gray!25, shape=ellipse, inner ysep=1.6cm, inner xsep=1cm},
link/.style={-,shorten >=1pt,semithick}
]
\node[name=R,shape=rectangle, inner xsep=3.5cm] {};
\node[name=T,elips] at (R.west)[left=15] {};
\node[place] (w1) at (T.mid) [left=10, below=35] {$a$};
\node[place] (w2) at (T.mid) {$b$};
\node[place] (w3) at (T.mid) [right=10, above=40] {$c$};
\draw[post] (w3) to (w2);
\draw[dashed,thick] (w1) to (w2);

\node[name=T1,elips] at (R.west)[right=15] {};
\node[place] (w1) at (T1.mid) [left=10, below=35] {$a$};
\node[place] (w2) at (T1.mid) {$b$};
\node[place] (w3) at (T1.mid) [right=10, above=40] {$c$};
\draw[post] (w2) to (w1);
\draw[dashed,thick] (w3) to (w2);

\draw[arrow] (T) + (1.6,0) to (T1)[left=3];
r(
\matrix[column sep=1cm] at (R.east)
{
\node {.}; & \node{.}; & \node{.};  \\
\node {.}; & \node{.}; & \node{.}; \\
\node {.}; & \node{.}; & \node{.}; \\
\node {.}; & \node{.}; & \node (fw3){}; \\
\node {.}; & \node{.}; & \node{.}; \\
\node {.}; & \node{.}; & \node{.}; \\
\node (fw2){}; & \node (fw1){}; & \node{.};\\
\node {.}; & \node{.}; & \node{.}; \\
\node {.}; & \node{.}; & \node{.}; \\
\node {.}; & \node{.}; & \node{.}; \\
};
\node (ww2) at (fw2)[left] {$b$};
\node (ww1) at (fw1)[right] {$b$};
\node (ww3) at (fw3)[right] {$a$};
\draw[link] (fw2)circle(0.7mm) to (fw1)circle(0.7mm);
\draw[link] (fw1)circle(0.7mm) to (fw3)circle(0.7mm);
\end{tikzpicture}

\caption{Edge slidings and associated flow.}
\label{slidings}
\end{center} 
\end{figure}
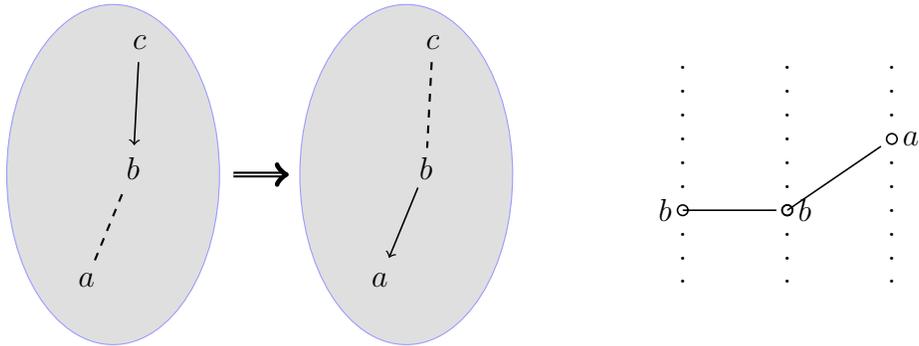 


\end{enumerate} 
How do these moves correspond to flows ? 
\begin{enumerate} 
\item The edge reversal $(w_{1},w_{2})$ corresponds (Figure~\ref{reversals}) to one unit of flow from node $w_{2}$ to node $w_{1}$ on the next level. 
\item The rotation $(w_{1},w_{2},w_{3})$ as described above corresponds (Figure~\ref{rotations}) to one unit of flow from node $w_{1}$ to node $w_{1}$ on the next level. 
\item An edge sliding (Figure~\ref{slidings}) involves using two levels of the flow network. 
\end{enumerate} 

It is easy to see that flows associated to biased edge reversals and rotations, 
or to preserving an edge (and its orientation) satisfy the biased flow condition. 
Hence (by composability of biased flows) this holds for biased slidings as well. 

It remains to show that we can transform $T_{OPT}$ to $T_{GREEDY}$ using biased moves of type 1,2 and 3. The strategy comprises three 
parts, described by  the following: 
\begin{itemize} 
\item[(a).] Consider all edges $e\in T_{OPT}\cap T_{GREEDY}$, oriented in 
the same way in both trees. We need to do nothing about them. 
\item[(b).] Consider all edges $e=(a,b)\in T_{OPT}\cap T_{GREEDY}$, with 
opposite orientations in the two trees. We will reorient them by an edge reversal. 
\item[(c).] Consider all edges in $T_{OPT}\setminus T_{GREEDY}$. 
We will iteratively replace such edges with edges in $T_{GREEDY}\setminus T_{OPT}$,
in a way such that the resulting intermediate graphs are in fact trees.  

The strategy for iterative replacement employs the current tree, denoted by $T_{1}$.
 Initially $T_{1}=T_{OPT}$. Let us consider an edge $e=(a,b)\in T_{1}\setminus T_{GREEDY}$.
 $e$ is in fact in $T_{OPT}\setminus T_{GREEDY}$, as edges added in the iterative 
process belong to $T_{GREEDY}$. Assume without loss of generality that $rank(a)>rank(b)$. 
Since $e\in E$ and $e\not \in T_{GREEDY}$, $a$ is connected to a node $c$ with $rank(c)<rank(b)$
in the same component to $b$ (thus creating a cycle that would preclude adding edge $e$).
$c$ is in fact the neighbor of $a$ on the unique path towards the root of $T_{GREEDY}$. 

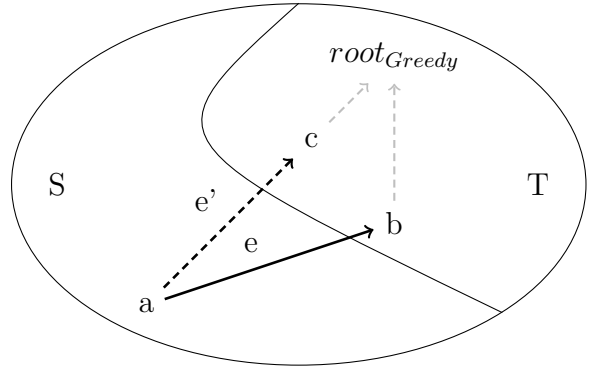
\begin{figure}[h!]
\begin{center}
\begin{tikzpicture}
[elips/.style={shape=ellipse,draw , inner ysep=1.7cm, inner xsep=2.7cm},
linkopt/.style={->,shorten >=0.5pt,draw=gray!50, thick},
linkG/.style={->,shorten >=0.5pt, densely dashed, draw=gray!50, thick},
link/.style={-,dotted, draw=gray!100}
]

\node[name=E,elips]{};

\matrix [
   column sep={1.1cm,between origins},
   row sep={1.1cm,between origins}, 
   ] at (E.mid)
{
  & & & \node (n3) {$root_{Greedy}$}; \\
  \node (n1) {}; & & \node (nc) {c};  & \\
  & & & \node (n8) {b}; \\
  \node (n12) {a}; & & & \\
};

\draw[linkopt, line width=1pt, draw=black!100] (n12) to node[auto]{e} (n8);
\draw [linkG] (n8) to (n3);
\draw [linkG] (nc) to (n3);
\draw [linkG, line width=1pt, draw=black!100] (n12) to node[auto]{e'} (nc);

\draw (E.north) .. controls (n1) .. (E.south east);
\node (T) at (E.east)[left=10]{T};
\node (S) at (E.west)[right=10]{S};
\end{tikzpicture}
\end{center}
\caption{Iterative transformation of edges from $T_{OPT}\setminus T_{GREEDY}$ into edges in $T_{GREEDY} \setminus T_{OPT}$ in Lemma~\ref{first-lemma}. First the edge $e$ is reoriented towards $a$. Then we slide it into $e^{\prime}$.}
\end{figure}

Eliminating $e$ from $T_{1}$ breaks down the set of vertices into two disjoint 
connected components $S$ and $T$, with endpoints $a,b$ into disjoint components. 
Together with 
the edges of $T_{GREEDY}$, edge $e$ determines an unique cycle $C$, consisting of the 
edges on the path from the root towards $a$ and $b$, respectively, plus edge $e$. 
There exists, 
therefore an edge $e^{\prime}\neq e$ on this cycle $C$, whose endpoints are one in 
$S$ and one in $T$. We infer the fact that $e^{\prime}\in T_{GREEDY}\setminus 
T_{1}$.  

 If $e^{\prime}$ is on the path from $a$ to the root of $T_{GREEDY}$ then we can use slidings to eliminate edge $e$ from the tree and add edge $e^{\prime}$ to the tree instead. We 
may also need to perform the reversal of edge $e$ before we can make the sliding (in case that edge $e$ is oriented towards $b$ in $T_{OPT}$). But since $rank(c)<rank(b)$ {\em all resulting 
flows (including the one corresponding to reorienting edge $e$ and then sliding it) are 
biased}.

If on the other hand $e^{\prime}$ is on the path from $b$ to the root of $T_{GREEDY}$ then we first use a greedy edge reversal (possible, as $rank(a)>rank(b)$), then edge slidings to 
replace $e^{\prime}$ by $e$. In both cases, crucially {\em   all resulting 
flows (including the one corresponding to reorienting edge $e$ and then sliding it) are 
biased}.

We only have to show that the resulting graph $T_{1}^{\prime}=T_{1}\setminus e+e^{\prime}$ is a tree (acyclic), so that the invariant is respected. Indeed, $e^{\prime}$ has its endpoints in $S$ and $T$, respectively, and is the unique edge of $T_{1}^{\prime}$ with this property. Therefore it is part of no cycle in $T_{1}^{\prime}$. Since $T_{1}$ was acyclic, $T_{1}^{\prime}$ is acyclic too (hence a tree).
\end{itemize} 

Each edge move of one of the three types above corresponds to a distinct path in the flow network, described as follows: 
\begin{itemize} 
\item[(a).] Edges $e$ shared (with the same orientation, say towards node $j$) between $T_{OPT}$ and $T_{GREEDY}$ correspond to paths between vertices with the same index $j$. 
\item[(b).] Reorienting an edge $e=(j,l)$ from $j$ towards $l$ corresponds to sending one unit of flow from node $j$ on the first level to node $l$ on the next one, and then routing that unit of flow across nodes with label $l$. 
\item[(c).] Moves of type [c.] correspond to flows in a similar way, except that they might involve multiple edges (to comply with capacity constraints), and thus multiple nontrivial steps. 
As argued, though, above, all resulting flows are biased. 
\end{itemize}  
\qed

We exemplify the transformation from the previous lemma in the example from Figure~\ref{biasedflow}. The graph $G$ consists of three nodes, considered in the order $rank(a)<rank(b)<rank(c)$ by the GREEDY algorithm. To go from the optimal solution to the greedy one we first reverse orientation on the edge (a,b). This corresponds to one unit of flow from node $b$ to node $a$ (and subsequently being routed to nodes labeled $a$). The second transformation consists of first performing an edge reversal on edge $(b,c)$ and then sliding edge $(b,c)$ towards $a$. 
The associated flow goes from $b$ to $a$, going through nodes labeled $c$, exemplifying the fact that the biased condition is only valid at the extremities of the flow. 

\begin{figure} 
\begin{center} 
\begin{tikzpicture}
[pre/.style={<-,shorten <=1pt,semithick},
place/.style={thick,inner sep=1pt,minimum size=5mm, right},
post/.style={->,shorten >=1pt,semithick},
arrow/.style={->, aspect=4, double, shorten >=4pt, thick},
elips/.style={shape=ellipse, inner ysep=1.6cm, inner xsep=0.4cm},
link/.style={-,shorten >=1pt,semithick}
]
\node[name=R,shape=rectangle, inner xsep=2cm] {};
\node[name=T1,elips] at (R.west)[left=85] {};
\node[place] (w1) at (T1.mid) [left=10, below=30] {$b$};
\node[place] (w2) at (T1.mid) [right=7, above=35] {$a$};
\node[place] (w3) at (T1.mid) [right=24, below=5] {$c$};
\node[place] at (T1.south){OPT};
\draw[post] (w2) to (w1);
\draw[post] (w3) to (w1);
\draw[dashed,thick] (w2) to (w3);

\node[name=T2,elips] at (R.west)[left=15] {};
\node[place] (w1) at (T2.mid) [left=10, below=30] {$b$};
\node[place] (w2) at (T2.mid) [right=7, above=35] {$a$};
\node[place] (w3) at (T2.mid) [right=24, below=5] {$c$};
\draw[post, thick, draw=blue!60] (w1) to (w2);
\draw[post] (w3) to (w1);
\draw[dashed,thick] (w2) to (w3);

\draw[arrow] (T1) + (1, 0) to (T2);

\node[name=T3,elips] at (R.west)[right=15] {};
\node[place] (w1) at (T3.mid) [left=10, below=30] {$b$};
\node[place] (w2) at (T3.mid) [right=7, above=35] {$a$};
\node[place] (w3) at (T3.mid) [right=24, below=5] {$c$};
\draw[post] (w1) to (w2);
\draw[post, thick] (w1) to (w3);
\draw[dashed,thick] (w2) to (w3);

\draw[arrow] (T2) + (1, 0) to (T3);

\node[name=T4,elips] at (R.west)[right=85] {};
\node[place] (w1) at (T4.mid) [left=10, below=30] {$b$};
\node[place] (w2) at (T4.mid) [right=7, above=35] {$a$};
\node[place] (w3) at (T4.mid) [right=24, below=5] {$c$};
\node[place] at (T4.south){GREEDY};
\draw[post, thick, draw=blue!60] (w1) to (w2);
\draw[post, thick, draw=blue!60] (w3) to (w2);
\draw[dashed,thick] (w1) to (w3);

\draw[arrow] (T3) + (1, 0) to (T4);

\matrix[column sep=0.5cm] at (R.east)[right=50]
{
\node (a1){a}; & \node(a2){a}; & \node(a3){a}; & \node(a4){a}; \\
\node (b1){b}; & \node(b2){b}; & \node(b3){b}; & \node(b4){b};\\
\node (c1){c}; & \node(c2){c}; & \node(c3){c}; & \node(c4){c};\\
};
\draw[link] (b1) to (a2);
\draw[link] (a2) to (a3);
\draw[link] (a3) to (a4);
\draw[link] (b1) to (c2);
\draw[link] (c2) to (c3);
\draw[link] (c3) to (a4);
\end{tikzpicture}
\end{center} 
\caption{(a). Transfoming the greedy to the optimal solution using the moves from Lemma~\ref{first-lemma}. (b). The associated biased flow.}
\label{biasedflow} 
\end{figure}
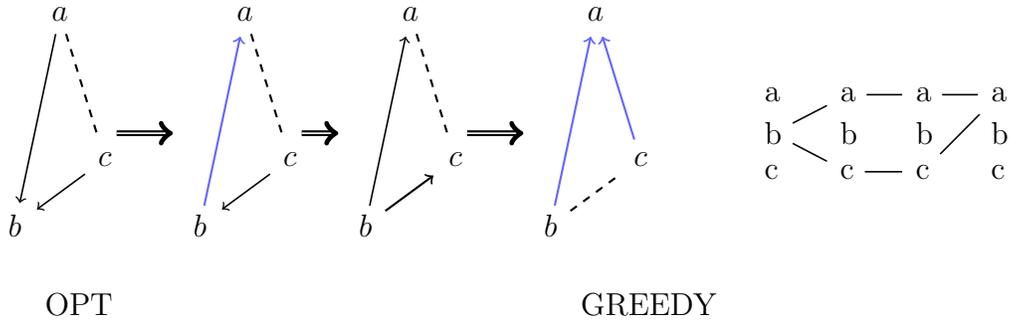 

\begin{definition} 
Let $<$ be any total path ordering such that: 
\begin{enumerate} 
\item All paths $(j,s)$, $j\neq s$ come before all paths of type $(p,p)$. 
\item Among paths of the first type $P_{i}=(j_{i},l_{i})$, $i=1,2$, $l_{1}<l_{2}\Rightarrow P_{1}<P_{2}$. 
\end{enumerate} 
\end{definition} 

\begin{lemma}\label{second-lemma} 
The flow $f$ constructed in the proof of Lemma~\ref{first-lemma} is admissible with respect to $<$. 
\end{lemma} 
Consider a path $P$ between nodes $j$ and $i_{m}$,  such that $f(P)>0$. There are two cases: \\
\noindent{\bf Case 1: $j\neq i_{m}$}.\\Then, by the biased nature of the flow, $j\not \in W_{m-1}$, that is node $j$ is a candidate for the greedy algorithm at stage $m$. 
Since $i_{m}$ was chosen instead, the number of edges that would be oriented 
towards $j$, should it be chosen at stage $m$, is less or equal to the number of edges oriented towards $i_{m}$ at that stage. The second quantity is (by the definition of the GREEDY algorithm) nothing but $Y_{i_{m}}$. 

To interpret the first quantity, we will associate edges in $T_{OPT}\setminus T_{GREEDY}$ to paths of unit flow starting from $j$, in such a way that all edges mapped to some path $Q$, $P\leq Q$, could be oriented towards $j$ should this node be chosen at stage $m$. 

First, note that every possible edge $(j,l)$ with $l\in W_{m-1}$, oriented towards $j$ in the spanning tree $T_{OPT}$  is either present in $T_{GREEDY}$ (but necessarily oriented towards $l$) or has been replaced (using the process of step [c]) by an edge rooted at some 
vertex $v$ of even lower rank than $l$. Thus the edge corresponds to a one unit of flow on 
some path $Q$ between $j$ and $l$ (or $v$), a path that is lexicographically smaller than 
$P$. 

Similarly, if $l$ is not itself in $W_{m-1}$ but is in $\delta(W_{m-1})$, and 
$j$ is connected to a node in the same connected component as $l$ 
(after step $m-1$ of the GREEDY algorithm) , then edge $(j,l)$ 
cannot be in $T_{GREEDY}$, under any orientation (or else it would create a cycle $C$). 
When considered by the GREEDY algorithm it is swapped under step [c.]. Note that cycle $C$ 
(except edge $(j,l)$) is contained in $W_{m-1}\cup \delta(W_{m-1})$, with every edge in this 
cycle being assigned to a vertex in $W_{m-1}$. Hence the edge also corresponds to a one unit of flow on 
some path $Q$ between $j$ and $l$ (or $v$), a path that is lexicographically smaller than 
$P$. 

Thus any unit of flow from node $j$ sent on a path that is scheduled after path $P$ corresponds to some 
edge $(j,l)$ not covered by one of the previous two cases. All remaining such edges are among those available for $j$ at stage $m$, were it to be chosen by 
the GREEDY algorithm. Their number is, as we saw, at most $Y_{i_{m}}$, the flow into 
$i_{m}$ in the GREEDY solution.
\\
\noindent{\bf Case 2: $j=i_{m}$}.\\ This is trivial, as $P$ is the only path leaving node $j$ at this stage (and is among those that arrive at $j$).

Hence the flow is admissible. 
\qed

To complete the proof of Theorem~\ref{mest}, we simply apply Proposition~\ref{thm-main-two}. 
\qed

\section{Application to Cooperative Games} 
\label{games} 

The setting in this paper has an alternative reformulation using the language of cooperative game theory \cite{branzei2008models}. Indeed, a  problem in this area is our main target application and originally motivated our research. It forms the subject of a companion paper \cite{istrate-bonchis2012-tugames}. Here we only  provide a brief outline of our approach. 

A {\em transferable utility (TU) coalitional (cost) game} consists of a finite set of players $U$ and a monotone function (called 
{\em characteristic cost function}) $c:\mathcal{P}(U)\rightarrow {\bf R}$ that satisfies $c(\emptyset)=0$.  
A set $T\subseteq U$ is called a {\em coalition}, with $T=U$ called {\em the grand coalition}. A TU game is called {\em concave}  if its cost function is submodular. Denote $N$ the cardinal of $U$. A {\em cost allocation (imputation)} is simply a vector $x=(x_{1},\ldots, x_{N})\in {\bf R} ^{N}_{+}$.

 The notion of rationality of a cost allocation is embodied by the {\em core of a cooperative game} $G=(U,c)$ defined as the set of cost allocations that satisfy the following conditions: 
\begin{enumerate} 
\item{\bf Efficiency:} $\sum_{i\in U} x_{i}=c(U)$. 
\item{\bf Individual rationality:}  $0\leq x_{i}\leq c(\{i\})$, for all $i\in U$. 
\item{\bf Coalitional rationality:} 
\[
\sum_{i\in S}x_{i}\leq c(S)\mbox{, } \forall\mbox{ } \emptyset \neq S\subseteq U; 
\]
\end{enumerate} 

Clearly core elements in a concave game can be regarded as solutions to the corresponding instance of submodular set cover. 

Allocations in the core may be seen as ``rational'', but they need not necessarily be ``fair''. The classical approach to fairness in cost allocations is axiomatic, and identifies the  celebrated {\em Shapley Value} \cite{roth1988shapley} as 
the unique cost allocation satisfying several natural conditions. However, despite its intrinsic attractiveness and conceptual power, the Shapley value may be inappropriate as a ``fair'' solution concept for many reasons, including the setting of coalitions with a 
dynamic structure 
, in games (necessarily not concave) 
for which the Shapley value  is not in the core, or in the presence of social preferences in favor of other social arrangements (e.g. egalitarianism \cite{dutta1989concept}). 

This is not to say that any particular alternative to the Shapley value would 
fare better: for one, any such measure would violate at least one of the 
axioms that uniquely identify the Shapley value, and may have other drawbacks. 

The situation is somewhat similar to the issues arising in non-cooperative settings with respect to Nash equilibria, the canonical concept of rationality in that setting. Just as the core of a cooperative game may be very large, a non-cooperative game may have multiple Nash Equilibria. Some of these equilibria may be ``suboptimal'' the selfish nature of goals pursued by individual agents may lead to suboptimal system performance. As long as we cannot exclude some of these equilibria there is no way to rule out suboptimal behavior. 

The seminal work of Roughgarden and Tardos \cite{roughgarden2002bad},\cite{selfish-routing} has opened an interesting avenue in dealing with the multiplicity of equilibria in noncooperative settings. Instead of attempting to propose a normative solution to equilibrium selection, their price of anarchy measure takes a worst-case perspective, quantifying the degradation in system performance due to uncoordinated behavior, measured on the {\em worst} strategy profile still compatible with individual rationality. 

The objective of paper \cite{istrate-bonchis2012-tugames} is to propose approach with a similar philosophy for cooperative games. Rather than attempting to postulate any particular ``fair solution'' of a cooperative game, 
we will investigate the fairness of an arbitrary allocation in the core. 

Fairness will be measured with respect to a ``baseline'' cost distribution $q$, deemed ``reasonable''.  
Given an arbitrary cost distribution $p$ we will use the {\em Shannon divergence} $D(p||q)=\sum_{i\in I} p_{i}\log_{2}(p_{i}/q_{i})$ to measure the ``distance'' of 
distribution $p$ with respect to the baseline distribution $q$. Note that $D$ is not a metric (as it does not satisfy the triangle inequality) but is a {\em pseudo-metric} and has been employed before as a ``distance'' between two distributions. 

Given a cooperative game $G$, its {\em worst-case fairness} with respect to cost allocation $q$ is defined as the supremum of $D(p||q)$ over all distributions $p$ arising from an imputation in $core(G)$. 

Depending on the way to select $q$ we may have several versions of the worst-case fairness measures, including the following ones:   
\begin{itemize}
\item {\em strictly egalitarian:} $q$ is the uniform distribution $q_{i}=1/N$. 
\item {\em egalitarian:} $q$ is the egalitarian solution of Dutta and Ray \cite{dutta1989concept}. 
\item {\em marginalist:} $q$ is the cost distribution corresponding to the Shapley value.  
\end{itemize} 

The strictly egalitarian worst-case fairness can be directly related to the setting of this paper, as in that case Shannon divergence directly relates to entropy. Certainly this seting is the the most controversial in terms of applicability, though strict egalitarianism as an approach has featured in a substantive manner in economics and other social sciences.
One can further justify its study on grounds related to {\em mechanism design}: suppose the cooperative game is not ``given'' but can be imposed on the set of players. Viewed this way strictly egalitarian worst-case fairness is a 
measure of the {\em design}, rather than the resulting cost allocation. 

On the other hand the connection between divergence and entropy extends (perhaps in a more limited setting) to marginalist approaches too:  for certain games (including the {\em induced subgraph games} from \cite{deng1994complexity}) one can quantify the performance of the GREEDY and other approximation algorithms to the marginalist worst-case fairness. Once again we refer to \cite{istrate-bonchis2012-tugames}, where full details will be provided. 

\section{Conclusions}

Obtaining a tight result for entropy minimization problem on matroids remains a topic 
for further research. We believe that $\beta = 1$ at least for a large class of particular versions of this problem. Formulating and proving such a result remains still open, though. Even in the case of MEST we don't know whether our result is optimal. 

The game theoretic investigations opened by our results have multiple variations: notions of ``worst-case fairness'' can be investigated for a variety of combinatorial games \cite{branzei2008models}, for various other solution concepts such as the $\epsilon$-core, the least core, the kernel 
or the nucleolus.

Finally, as mentioned, the problem we studied has some promising potential applications.
It would be interesting to develop these applications.

\section*{ Acknowledgments} 

Both authors contributed equally to this work. 

The first author has been supported by a project on Postdoctoral national programs, contract CNCSIS PD\_575/2010. 

The corresponding author  has been supported by a project on Postdoctoral programs for sustainable development in a knowledge-based society, contract POSDRU/89/1.5/S/60189, cofinanced by the European Social 
Fund via the Sectorial Operational Program for Human Resource Development 2007-2013. 

\section*{Appendix: Proof of Theorem~\ref{thm:hard}}
 
We will use an idea related to the strategy employed to 
prove the NP-completeness of the Minimum Labeling Spanning Tree 
Problem \cite{chang1997minimum}. We will provide a reduction from the NP-complete problem Minimum Entropy Set Cover  to MEST. 

Indeed, let $M=(U,\mathcal{P})$ be an instance of Minimum Entropy Set Cover problem, with $U=\{1,2,\ldots, n\}$ and $\mathcal{P}=\{P_{1},P_{2},\ldots, P_{m}\}$. 

Define a graph $G_{M}$ as follows: 
\begin{enumerate} 
\item $G_{M}$ has one super-node $R$, $m+n-1$ auxiliary nodes $A_{1},\ldots A_{m+n-1}$ (connected only to $R$), $m$ nodes corresponding to sets $P_{1},P_{2},\ldots, P_{m}$
and $n$ nodes corresponding to elements $1,2,\ldots n$, respectively. 
\item Nodes corresponding to $P_{i}$ are connected to $R$ and to nodes corresponding to elements $j$, $j\in P_{i}$. 
\item These are all edges of $G_{M}$. 
\end{enumerate}

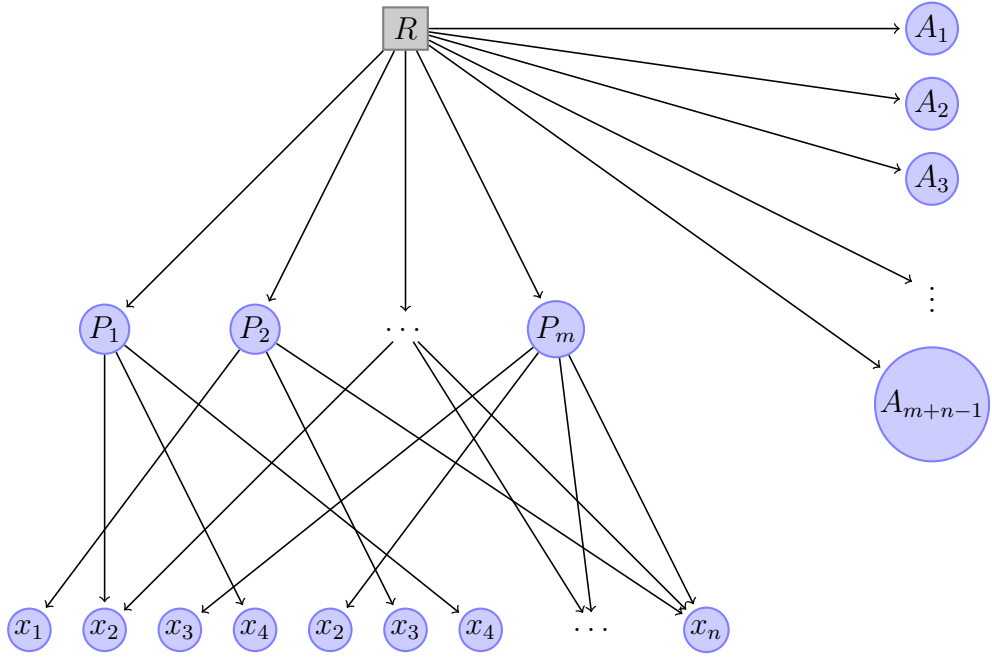
\begin{figure}[h]
\begin{center} 
\begin{tikzpicture}
[place/.style={circle,draw=blue!50,fill=blue!20,thick,inner sep=1pt,minimum size=4mm},
transition/.style={rectangle,draw=black!50,fill=black!20,thick},
pre/.style={<-,shorten <=2pt,semithick},
post/.style={->,shorten >=2pt,semithick}]
\node[transition] (root) at 	(5,8) {$R$};
\node[place] (p1) at 	(1,4) {$P_1$};
\node[place] (p2) at 	(3,4) {$P_2$};
\node[] (p3) at 	(5,4) {$\dots$};
\node[place] (pm) at 	(7, 4) {$P_m$};
\node[place] (x1) at 	(0,0) {$x_1$};
\node[place] (x2) at 	(1,0) {$x_2$};
\node[place] (x3) at 	(2,0) {$x_3$};
\node[place] (x4) at 	(3,0) {$x_4$};
\node[place] (x5) at 	(4,0) {$x_2$};
\node[place] (x6) at 	(5,0) {$x_3$};
\node[place] (x7) at 	(6,0) {$x_4$};
\node[] (x8) at 	(7.5,0) {$\dots$};
\node[place] (xn) at 	(9, 0) {$x_n$};

\node[place] (A1) at 	(12, 8) {$A_1$};/
\node[place] (A2) at 	(12, 7) {$A_2$};/
\node[place] (A3) at 	(12, 6) {$A_3$};/
\node[] (A4) at 	(12, 4.5) {$\vdots$};/
\node[place] (A5) at 	(12, 3) {$A_{m+n-1}$};/ 
\foreach \x in {p1, p2, p3, pm, A1, A2, A3, A4, A5}
	\draw[post] (root) to (\x);
\foreach \x in {x2, x4, x7}
	\draw[post] (p1) to (\x);
\foreach \x in {x1, x6, xn}
	\draw[post] (p2) to (\x);
\foreach \x in {x3, xn, x5, x8}
	\draw[post] (pm) to (\x);
\foreach \x in {x2, x8, xn}
	\draw[post] (p3) to (\x);
\end{tikzpicture}
\end{center}
\caption{Graph $G_{M}$ in the construction from Theorem~\ref{thm:hard}}
\label{mesc-graph} 
\end{figure} 

To relate the minimum entropy spanning tree on $G_{M}$ and the minimum cover on $M$ we need the following 

\begin{claim} 
  Let $1\leq a\leq b\leq a+b\leq W$. Then 
\[
-\frac{a}{W}\log_{2}(\frac{a}{W})-\frac{b}{W}\log_{2}(\frac{b}{W})\geq -\frac{a-1}{W}\log_{2}(\frac{a-1}{W})-\frac{b+1}{W}\log_{2}(\frac{b+1}{W}).
\]
\end{claim}
 
{\bf Proof. } 
This is equivalent to 
\[
\frac{a-1}{W}\log_{2}(\frac{a}{a-1})+\log_{2}(\frac{a}{W})\leq \frac{b}{W}\log_{2}(\frac{b+1}{b})+\log_{2}(\frac{b+1}{W}).
\]
or
\[
\frac{1}{W}\log_{2}[(1+\frac{1}{a-1})^{a-1}]+\log_{2}(\frac{a}{W})\leq \frac{1}{W}\log_{2}[(1+\frac{1}{b})^{b}]+\log_{2}(\frac{b+1}{W}).
\]
This follows easy from the monotone increasing nature of function $g(x)=(1+\frac{1}{x})^{x}$. 
\qed

So let us consider a spanning tree $T_{M}$ in $G_{M}$ of minimum entropy. $T_{M}$ has to contain edges between $R$ and $A_{i}$ (as they are the unique edge containing vertex $A_{i}$). Moreover, by a simple application of the claim, we may assume without loss of generality that edge $(T,A_{i})$ in the minimum entropy solution is contributed by vertex $T$, who has degree at least $m+n-1$ from the auxiliary nodes only, thus larger or equal to that of nodes $P_{1},\ldots, P_{m}$, in the spanning tree $T_{M}$.

Assume now, for the sake of contradiction, that some node $P_{i}$ would be a node unconnected to $R$ in $T_{M}$. Thus $P_{i}$ is connected to some non-leaf node $j$. Deleting edge $P_{i},j$, adding edge $R,P_{i}$ (contributed by $R$) and taking into account that the degree of node $j$ in $T_{M}$ is at most $m$ we would get a tree of lower entropy. 

The conclusion of this argument is that each node $P_{i}$ is connected to $R$ in $T_{M}$, 
with edge $(R,P_{i})$ contributed by $R$. 

From this conclusion it follows easily that every node $j$ is connected in $T_{M}$ to at most one $P_{i}$ (or else $T_{M}$ would have a cycle), thus corresponding to a cover $D$ in $M$. Moreover, to be a minimum entropy cover $C$ of $T_{M}$ we may assume that each such edge is contributed by node $P_{i}$. 

To compute the entropy of cover $C$ of $T_{M}$ we first consider the contribution of node $R$, equal to 
\[
-\frac{2m+n-1}{2(m+n)}\log_{2}[\frac{2m+n-1}{2(m+n)}]
\] 
Assuming node $P_{i}$ has degree $d_{i}$ in cover $C$, the contribution of such nodes to the entropy of cover $C$ is  
\small{
\begin{align*}
 & -\sum_{i=1}^{m}\frac{d_{i}}{2(m+n)}\log_{2} \frac{d_{i}}{2(m+n)} =-\frac{1}{2(m+n)}\left[\sum_{i=1}^{m}d_{i}\left(\log_{2}d_{i}-\log_{2}2(m+n)\right)\right]=\\
 & =-\frac{1}{2(m+n)}\left[\sum_{i=1}^{m}d_{i}\log_{2}(d_{i})-n\log_{2}2(m+n)\right]=\\
 & =\frac{n}{2(m+n)}\left[-\sum_{i=1}^{m}\frac{d_{i}}{n}\log_{2}\left(\frac{d_{i}}{n}\right)\right]+\frac{n}{2(m+n)}\log_{2}\frac{2(m+n)}{n}.\end{align*}
}
Thus 
\begin{align*}
Ent(C)  =  &-\frac{2m+n-1}{2(m+n)}\log_{2} \frac{2m+n-1}{2(m+n)} +\frac{n}{2(m+n)}\log_{2}\frac{2(m+n)}{n} +\\
       & +  \frac{n}{2(m+n)}\cdot Ent(D),
\end{align*}
in particular instance $M$ has a cover of entropy at most $\lambda$ if and only if instance $G_{M}$ of MEST has a cover of entropy at most 
\[
-\frac{2m+n-1}{2(m+n)}\log_{2}\left[\frac{2m+n-1}{2(m+n)}\right]+\frac{n\log_{2}(2+2m/n)}{2(m+n)}+\frac{n}{2(m+n)}\cdot \lambda. 
\]
\qed


\bibliographystyle{alpha} 
\bibliography{/home/gistrate/Dropbox/texmf/bibtex/bib/bibtheory}

%
%
%





\end{document}